\documentclass[sn-mathphys,Numbered]{sn-jnl}

\usepackage{threeparttable}

\usepackage{subfigure} 
\usepackage{ulem}
\usepackage{graphicx}%
\usepackage{multirow}%
\usepackage{amsmath,amssymb,amsfonts}%
\usepackage{amsthm}%
\usepackage{mathrsfs}%
\usepackage[title]{appendix}%
\usepackage{xcolor}%
\usepackage{textcomp}%
\usepackage{manyfoot}%
\usepackage{booktabs}%
\usepackage{algorithm}%
\usepackage{algorithmicx}%
\usepackage{algpseudocode}%
\usepackage{listings}%

\theoremstyle{thmstyleone}%
\newtheorem{thm}{Theorem}[section]
\newtheorem{proposition}[thm]{Proposition}%
\makeatletter

\newcommand{\Rmnum}[1]{\expandafter\@slowromancap\romannumeral #1@}
\makeatother
\usepackage{amsfonts}
\usepackage{txfonts}
\usepackage{paralist}
\usepackage{color}
\usepackage{tabu}                     
\usepackage{multirow}                 
\usepackage{multicol}                 
\usepackage{multirow}                
\usepackage{float}                    
\usepackage{makecell}                 
\usepackage{booktabs}                 

\usepackage{amsmath,bm}
\usepackage{amstext}
\usepackage{amsfonts}
\usepackage{amssymb}
\usepackage{graphicx}
\usepackage{verbatim}
\usepackage{bbm}
\newtheorem{theorem}{Theorem}[section]
\newtheorem{lemma}[theorem]{Lemma}
\newtheorem{remark}[theorem]{Remark}
\newtheorem{example}[theorem]{Example}
\newtheorem{corollary}[theorem]{Corollary}

\theoremstyle{thmstyletwo}%
\theoremstyle{thmstylethree}%
\usepackage{graphicx} 
\usepackage{subfigure} 
\usepackage{ulem}
\newtheorem{definition}{Definition}%

\begin{document}

\title[Article Title]{Test for high-dimensional linear hypothesis of mean vectors via random integration}

\author[1]{\fnm{Jianghao} \sur{Li}}

\author[2]{\fnm{Shizhe} \sur{Hong}} 
\author*[3]{\fnm{Zhenzhen} \sur{Niu}} \email{niuzz@sdnu.edu.cn}


\author[1]{\fnm{Zhidong} \sur{Bai}}

\affil[1]{\orgdiv{School of Mathematics and Statistics}, \orgname{Northeast Normal University}, \orgaddress{\street{5268 Peoples Road}, \city{Changchun}, \postcode{130024}, \country{China}}}

\affil[2]{\orgdiv{School of Statistics and Management}, \orgname{ Shanghai University of Finance and Economics}, \orgaddress{\street{No. 777, Guoding Road}, \city{Shanghai}, \postcode{200433}, \country{China}}}

\affil[3]{\orgdiv{ School of Mathematics and Statistics}, \orgname{Shandong Normal University}, \orgaddress{\street{ No.1 University Road,
			Science Park, Changqing District}, \city{Jinan}, \postcode{250358}, \country{China}}}

\abstract{
In this paper, we investigate hypothesis testing for the linear combination of mean vectors across multiple populations through the method of random integration. We have established the asymptotic distributions of the test statistics under both null and alternative hypotheses. Additionally, we provide a theoretical explanation for the special use of our test statistics in situations when the nonzero signal in the linear combination of the true mean vectors is weakly dense. Moreover, Monte-Carlo simulations are presented to evaluate the suggested test against existing high-dimensional tests. The findings from these simulations reveal that our test not only aligns with the performance of other tests in terms of size but also exhibits superior power.
}



\maketitle

\section{Introduction}
\label{intro}
High-dimensional data has grown more commonplace due to the quick development of data gathering and technological advancements, which are now present in many scientific fields, including finance, medicine, genomics, and other fields. The common characteristic of these high-dimensional datasets is that their dimensions are much greater than the number of samples, referred to as ``large $p$, small $n$". In such cases, traditional statistical techniques are either entirely useless or inapplicable. This has prompted statisticians to explore novel approaches to dealing with high-dimensional data, e.g., the widely used random matrix method to tackle large-dimensional multivariate statistical issues.

In recent years, the test for
mean vectors under high-dimensional setting have been a very hot topic in the literature because of its important applications. Many scholars, including
\cite{bai1996effect}; \cite{srivastava2008test}; \cite{chen2010two}; \cite{lopes2011more}; \cite{cai2014two}; \cite{xu2016adaptive}; \cite{zhangjt2020simple} and among others, are interested in the two-sample high-dimensional mean test problem. Besides the above two-sample test problem, the multi-sample test problem is also concerned by many researchers,  see, e.g., \cite{schott2007test} \cite{cai2014high}; \cite{hu2017on}; \cite{li2024test} and the references therein.

However, if we only test the equality of several high-dimensional
mean vectors, that is,
$\mathrm{H}_0: \boldsymbol{\mu}_{1} = \boldsymbol{\mu}_{2} = \ldots = \boldsymbol{\mu}_{q},$
we may not effectively
infer the differences between several populations. When the heteroscedastic one-way MANOVA test is rejected, the relationship between the mean vectors may exhibit the following scenarios: e.g., $\boldsymbol{\mu}_{2} = 2\boldsymbol{\mu}_{3}$ or  $\boldsymbol{\mu}_{1} -3\boldsymbol{\mu}_{2} + \boldsymbol{\mu}_{3}=0$. In fact, these scenarios can be generalized as making inferences on a linear combination of $q$ mean vectors.

Suppose that random samples
$\{\mathbf{x}_{i1}, \mathbf{x}_{i2}, \cdots, \mathbf{x}_{in_i}\}\in \mathbb{R}^{p}$
follows from the $i$-th populations with mean vectors $\boldsymbol{\mu}_{i}$, possibly unknown and unequal covariance matrices $\boldsymbol{\Sigma}_{i}$ among populations, $i=1,2,\cdots,q$.
Our primary interest is to test the following linear hypothesis:
\begin{align}\label{test on linear hypothesis of means1}
\mathrm{H}_0:\sum_{i=1}^{q} \beta_{i} \boldsymbol{\mu}_{i}=\mathbf{0}   ~~\text{ vs. }  \mathrm{H}_1: \mathrm{~not~} \mathrm{H}_{0},
\end{align}
where $\beta_{1}, \ldots, \beta_{q}$  are given scalars. For the hypothesis testing in (\ref{test on linear hypothesis of means1}), \cite{nishiyama2013testing} transformed the $q$ samples into one sample by borrowing the successful idea from the multivariate Scheff\'{e}'s transformation \cite{scheffe1943solutions} (also known as Bennett's transformation \cite{bennett1950note}), then proposed a Dempster trace test. By transforming the samples using methods akin to \cite{nishiyama2013testing}, \cite{jiang2017likelihood} presented a test statistic similar to the F-matrix by employing the advanced tools from {\it Random Matrix Theory}. Although Bennett's transformation, which converts multi-sample Behrens-Fisher problems into one-sample testing problems, offers an alternative approach to making the original problem \eqref{test on linear hypothesis of means1} be solved in a different way, Bennett's transformation is widely recognized as a very flawed solution to the Behrens-Fisher problem in classical multivariate theory, owing to its notable loss in power.
 \cite{li2017test} proposed a test for multiple linear combinations of mean vectors with unequal covariance matrices, which generalizes the test problem (\ref{test on linear hypothesis of means1}). A U-statistic-based test for a general linear hypothesis testing issue under heteroscedasticity was obtained by \cite{zhou2017high}. In actuality, for this test problem (\ref{test on linear hypothesis of means1}), the test statistics given by \cite{li2017test} and \cite{zhou2017high} are essentially equivalent. Without requiring the common covariance matrix assumption,
\cite{zhang2022linear} extended the well-known Welch-Satterthwaite chi-squared approximation technique to the general linear hypothesis testing problem. Nevertheless, the existing literature for the hypothesis test (\ref{test on linear hypothesis of means1}) does not provide a solution when nonzero signals in the linear combination of the true mean vectors are more dense or weakly dense nonzero signals. Therefore, it makes sense to develop a new test procedure for high-dimensional data.

In this paper, with the aid of random integration, we present a weighted $L_2$-norm-based test for the hypothesis (\ref{test on linear hypothesis of means1}) under moderate conditions. With adjustable parameter settings, some existing tests, such as \cite{chen2010two, li2017test, zhou2017high, jiang022nonparametric}, become special instances in our proposed framework. The newly proposed high-dimensional test for the hypothesis (\ref{test on linear hypothesis of means1}), is non-parametric and purposefully avoids defining an explicit relationship between the dimension $p$ and the total sample size $n$. Furthermore, when nonzero signals in the linear combination of the true mean vectors are more dense or weaker dense, we demonstrate theoretically that the power of our newly proposed high-dimensional test behaves distinctively from certain existing tests for the hypothesis (\ref{test on linear hypothesis of means1}), such as \cite{li2017test} or \cite{zhou2017high}, \cite{zhang2022linear}. In addition, we also fully utilize numerical results in a variety of settings to support our main results. 

The rest of this paper proceeds as follows:
Our test procedure for the hypothesis (\ref{test on linear hypothesis of means1}) as well as the asymptotic distribution of our test statistic are presented in Section~\ref{section2.1} and \ref{section2.2}. In Section~\ref{section2.3}, 
We explore the asymptotic power of our test and briefly discuss its power performance when the covariance matrix of several populations is identical. Extensive simulation studies are presented in Section ~\ref{section3}. Finally, technical proofs of the main theorems
are outlined in Appendix \ref{Appendix}.

\section{Methodology and main results
}
\subsection{Construction of the test statistic}\label{section2.1}
We initiate the development of the test statistic for the hypothesis problem \eqref{test on linear hypothesis of means1} by defining $\sum_{i=1}^q \beta_i \boldsymbol{\mu}_i \triangleq \boldsymbol{\mu}$. This definition transforms  \eqref{test on linear hypothesis of means1} into the following equivalent form: 
\begin{equation}
\mathrm{H}_0:\boldsymbol{\mu}=\mathbf{0}   ~~\text{ vs. }  \mathrm{H}_1:\boldsymbol{\mu} \neq \mathbf{0}.
\label{equivalent test on linear hypothesis of means}
\end{equation}
Denote $\mathbf{x}= \sum_{i=1}^{q} \beta_{i} \mathbf{x}_{i}$, where $\mathbf{x}_i$ represents the $i$-th populations with mean vectors $\boldsymbol{\mu}_{i}$ and covariance matrices $\boldsymbol{\Sigma}_{i}$, $i=1,2,\cdots,q$.
Note that
\begin{align*}
\boldsymbol{\mu}= \mathbf{0} &\Leftrightarrow \mathbb{E} \mathbf{x}= \mathbf{0} \Leftrightarrow \boldsymbol{\theta }^{\top} \mathbb{E} \mathbf{x}= 0, \text { for any } \boldsymbol{\theta } \in \mathbb{R}^{p} \\
& \Leftrightarrow \mathbb{E} \left(\boldsymbol{\theta }^{\top} \mathbf{x}\right)=0, \text { for any } \boldsymbol{\theta } \in \mathbb{R}^{p} .
\end{align*}
A novel non-parametric technique called ``random integration'' suggested by \cite{jiang022nonparametric} turns the original issue \eqref{test on linear hypothesis of means1} into a one-dimension issue by introducing a random weight vector and integrating on its distribution. 
Therefore, to quantify the variations of a linear combination of mean vectors of several populations, one can employ the random integration approach described in \cite{jiang022nonparametric}. 
Accordingly, testing problem \eqref{equivalent test on linear hypothesis of means} can be further translated to determine whether
\begin{align}\label{Random integration technique}
\operatorname{RI}_{w}(\mathbf{x}) \triangleq \int \mathbb{E}^{2}\left(\boldsymbol{\theta }^{\top}\mathbf{x}\right) w(\boldsymbol{\theta }) d \boldsymbol{\theta }=0,
\end{align}
where  $w(\boldsymbol{\theta })$ represents any positive weighting function. 

We deduce that  $\boldsymbol{\mu}= \mathbf{0}$  if and only if $ \operatorname{RI}_{w}(\mathbf{x})=0$, as indicated by Equation \eqref{Random integration technique}. An explicit expression for evaluating \eqref{Random integration technique} with a proper weight function $w$ is given in the following lemma.
\begin{lemma}\label{linear mean test using random integration}
	Given that  $w(\boldsymbol{\theta })=\prod_{i=1}^{p} w_{i}\left(\theta_{i}\right)$ with each $w_{i}(\cdot)$ being a density function with a mean $ \alpha_{i}$  and variance  $\omega_{i}^{2}$  for  $i=1, \ldots, p$, then we have
	\begin{align*}
	\operatorname{RI}_{\theta}(\mathbf{x} ) & \triangleq \operatorname{RI}_{w}(\mathbf{x} ) \\
	& =\boldsymbol{\mu}^{\top} \mathbf{\Omega} \boldsymbol{\mu} + \left(\boldsymbol{\mu}^{\top} \boldsymbol{\alpha}\right)^{2} = ( \sum_{i=1} \beta_{i} \boldsymbol{\mu}_{i} )^{\top}\left( \mathbf{\Omega} +\boldsymbol{\alpha}\boldsymbol{\alpha}^{\top} \right) ( \sum_{i=1} \beta_{i} \boldsymbol{\mu}_{i} ),
	\end{align*}
	where $\boldsymbol{\alpha}=   \left(\alpha_{1}, \alpha_{2}, \ldots, \alpha_{p}\right)^{\top}$, and
	$$
	\mathbf{\Omega}=\left(\begin{array}{cccc}
	\omega_{1}^{2} & 0 & \cdots & 0 \\
	0 & \omega_{2}^{2} & \cdots & 0 \\
	\vdots & \vdots & \ddots & \vdots \\
	0 & 0 & \cdots & \omega_{p}^{2}
	\end{array}\right).  
	$$
	In addition, we have  $\operatorname{RI}_{\theta}(\mathbf{x} ) \geq 0$, achieving equality if and only if $ \sum_{i=1}^q \beta_i \boldsymbol{\mu}_i \triangleq \boldsymbol{\mu}= \mathbf{0}$.
\end{lemma} 
For notational convenience, denote $\mathbf{W} = \mathbf{\Omega} +\boldsymbol{\alpha}\boldsymbol{\alpha}^{\top}$. Again, the covariance matrix $\boldsymbol{\Sigma}_i$ and the mean vector $\boldsymbol{\mu}_i$ correspond to the $i$-th population $\mathbf{x}_{i}$, which is represented by the $i$-th sample $\mathbf{x}_{ij}$, $i=1,\dots,q,j=1,\dots,n_i$. To test the hypothesis problem \eqref{test on linear hypothesis of means1}, the following test statistic is taken into consideration as an unbiased estimator of the target function $( \sum_{i=1} \beta_{i} \boldsymbol{\mu}_{i} )^{\top} \mathbf{W} ( \sum_{i=1} \beta_{i} \boldsymbol{\mu}_{i} )$:
\begin{align}
T_{n}  =\sum_{i_{1} \neq i_{2}}^{q} \frac{\beta_{i_{1}} \beta_{i_{2}} }
{n_{i_{1}} n_{i_{2}}} \sum_{j_{1}, j_{2}}\left(\mathbf{x}_{i_{1} j_{1}}\right)^{\top} \mathbf{W} \left(\mathbf{x}_{i_{2} j_{2}} \right)
+\sum_{i=1}^{q} \frac{\beta_{i}^{2}}{n_{i}\left(n_{i}-1\right)} \sum_{j_{1} \neq j_{2}}\left(\mathbf{x}_{i j_{1}}\right)^{\top} \mathbf{W} \left(\mathbf{x}_{i j_{2}} \right).
\end{align}
\begin{remark}
	The exact form obtained with ``random integration'' is essentially a weighted $L_2$-norm test. The two-sample mean test, or $q=2$, $\beta_1 = 1$, $\beta_2 = -1$, is a specific example of hypothesis \eqref{test on linear hypothesis of means1}. Following that, \cite{jiang022nonparametric}'s test statistic is a particular instance of our test statistic. Specifically, when $\omega_{i} =1, \alpha_{i} = 0$, the expression for test statistic $T_{n}$ is reduced to U-statistics proposed by \cite{chen2010two}; Similarly, the test statistics proposed by \cite{li2017test} or \cite{zhou2017high} for the linear hypothesis \eqref{test on linear hypothesis of means1} are also included. In summary, the fact that it offers a unified framework with adjustable settings that makes it possible to include some of the current testing procedures is undoubtedly a benefit. In this article, for the sake of argument, we will confine ourselves to the case that $\alpha_{1} = \cdots = \alpha_{p}$,  $0<\omega_{1} \leq \cdots \leq \omega_{p}$. 
\end{remark}
Some straightforward calculations lead to 
\begin{align*}
\mathbb{E} \left( T_{n} \right)
= & \left( \sum_{i=1}^{q} \beta_{i} \boldsymbol{\mu}_{i} \right)^{\top} \mathbf{W} \left( \sum_{i=1}^{q} \beta_{i} \boldsymbol{\mu}_{i} \right),
\end{align*}
and $\operatorname{Var}\left(T_{n}\right)=\sigma_{n,q_1}^2+\sigma_{n,q_2}^2$ with
\begin{align*}
&\sigma_{n,q_{1}}^{2}=2 \sum_{i_{1} \neq i_{2}} \frac{\beta_{i_{1}} ^{2}\beta_{i_{2}}^{2}}{n_{i_{1}} n_{i_{2}}} \operatorname{tr}\left(  \mathbf{W} \boldsymbol{\Sigma}_{i_1}  \mathbf{W} \boldsymbol{\Sigma}_{i_2}\right)+2 \sum_{i=1}^{q} \frac{\beta_{i}^{4}}{n_{i}\left(n_{i}-1\right)} \operatorname{tr}\left( \mathbf{W} \boldsymbol{\Sigma}_{i}\right)^2,\\
&\sigma_{n,q_{2}}^{2}=4\left(\sum_{i=1}^{q} \beta_{i} \boldsymbol{\mu}_{i}\right)^{\top}  \mathbf{W}  \left(\sum_{i=1}^{q} \frac{\beta_{i}^2}{n_{i}} \boldsymbol{\Sigma}_{i}\right)  \mathbf{W}  \left(\sum_{i=1}^{q} \beta_{i} \boldsymbol{\mu}_{i}\right).
\end{align*}
The above results are postponed in the Appendix.

\subsection{The asymptotic null distribution of $T_n$}\label{section2.2}
To derive the asymptotic normality of our test statistic $T_n$, we impose a set of regular conditions. The key assumption is that the observations $\mathbf{x}_{ij},i=1,\dots,q,j=1,\dots,n_i$ are generated from the following structure:
\begin{align} \label{general factor model}
\mathbf{x}_{i j}=\boldsymbol{\mu}_i + \boldsymbol{\Gamma}_{i} \mathbf{z}_{i j}, i = 1,\ldots, q, j= 1, \ldots, n_i,
\end{align}
where $\boldsymbol{\Gamma}_{i}$ is a  $p \times m$  matrix with $m \geq p$ such that 
$\boldsymbol{\Gamma}_{i} \boldsymbol{\Gamma}_{i}^{\top}=\boldsymbol{\Sigma}_{i}$ ,  $\mathbf{z}_{i j}=\left( z_{ij1}, \ldots, z_{ijm   }\right)^{\top}$ are $m\times1$ independent and identically distributed (i.i.d.) random vectors satisfying $\mathbb{E}(\mathbf{z}_{i j})=\mathbf{0}$ and $\operatorname{Cov}(\mathbf{z}_{i j})=\mathbf{I}_m$, where $\mathbf{I}_m$ stands for the $m\times m$ identity matrix. This factor model structure is widely used in the literature, for instance \cite{bai1996effect}; \cite{chen2010tests} and the references therein. Additional assumptions are listed as follows:

\noindent{\textbf{Assumption A}}.
The components of $\mathbf{z}_{ij},i=1,\dots,q,j=1,\dots,n_i$ have uniformly bounded eighth moment, and satisfy
\begin{align}\label{moment conditions of z_ij}
\mathbb{E}\left( z_{i j \ell }\right)=0, \mathbb{E}\left( z_{i  j \ell }^2\right)=1, \mathbb{E}\left( z_{i  j \ell}^3 \right)=\kappa_{3i}\ \text{and}\ \mathbb{E}\left( z_{i  j \ell}^4 \right)=\kappa_{4i}
\end{align}
for $\ell \in \{1,\ldots, m \}$. Also, for positive integer  $v $ and  $\varsigma_{\ell}$  such that  $\sum_{\ell=1}^{v} \varsigma_{\ell} \leq 8$,
\begin{align}\label{pseudo-independence condition of z_ij}
\mathbb{E}\left(z_{i j \ell_{1}}^{\zeta_{1}} \cdots z_{i j \ell_{v}}^{\zeta v}\right)=\mathbb{E} \left(z_{i j \ell_{1}}^{\zeta_1}\right) \cdots \mathbb{E} \left(z_{i j \ell_{v}}^{\zeta_v}\right),
\end{align} 
whenever $\ell_{1}\neq \ell_{2}\neq \ldots \neq \ell_{v}$.

\noindent{\textbf{Assumption B}}. The sample sizes $n_i,i=1,\dots,q$ satisfy that $n_i/n \rightarrow k_{i} \in (0,1)$ for $i=1, \ldots q$ as  $n \rightarrow \infty$, where $n=\sum_{i=1}^{q} n_{i}$;

\noindent{\textbf{Assumption C}}. 
\begin{align}\label{bounded but spike existit}
\operatorname{tr}\left( \mathbf{W} \boldsymbol{\Sigma}_{i_1} \mathbf{W} \boldsymbol{\Sigma}_{i_2} \mathbf{W} \boldsymbol{\Sigma}_{i_3} \mathbf{W} \boldsymbol{\Sigma}_{i_4}\right)=
o\left[\operatorname{tr}^{2} \left\{\left(\sum_{i=1}^{q} \mathbf{W} \boldsymbol{\Sigma}_{i}\right)^{2}\right\}\right], 
\end{align} 
for $i_1, i_2, i_3, i_4 \in\{1,2, \ldots, q\} $;

\noindent{\textbf{ Assumption D}}.
\begin{align} \label{local alternative condition: test on linear hypothesis of means1}
\left(\sum_{i=1}^{q} \beta_{i} \boldsymbol{\mu}_{i}\right)^{\top}  \mathbf{W}  \left(\sum_{i=1}^{q} \beta_{i}^2\boldsymbol{\Sigma}_{i}\right)  \mathbf{W}  \left(\sum_{i=1}^{q} \beta_{i} \boldsymbol{\mu}_{i} \right) = o\left(n^{-1} \sum_{i_{1},i_{2}}^{q} \beta_{i_{1}}^{2} \beta_{i_{2}}^{2} \operatorname{tr}(  \mathbf{W} \boldsymbol{\Sigma}_{i_1}  \mathbf{W} \boldsymbol{\Sigma}_{i_2}) \right).
\end{align}


The above assumptions can be viewed as extensions of those in \cite{chen2010two} and \cite{li2017test}. Assumption A specifies the moment conditions and a pseudo-independence requirement for the components of $\mathbf{z}_{i j}$, which extends the case where $z_{ij\ell}$'s are i.i.d. with a bounded eighth moment. Assumption B explains that each population has a certain proportion of sample sizes in the total. Assumption C can be seen as the multi-sample version of conditions (2.4) in \cite{jiang022nonparametric}, which is achievable through appropriate selection of $\mathbf{W}$. Theoretically, consider a simple case that $\lambda_{1} \leq \lambda_{2} \cdots \leq   \lambda_{p}$  and  $\lambda_{1}^{*} \leq \lambda_{2}^{*} \cdots \leq \lambda_{p}^{*}$  are eigenvalues of  $\mathbf{W}$  and  $\boldsymbol{\Sigma}_{1} = \boldsymbol{\Sigma}_{2}= \cdots=\boldsymbol{\Sigma}_{q}=\boldsymbol{\Sigma}$, respectively, and assume that $\alpha_{1} = \cdots = \alpha_{p}=\alpha$,  $0<\omega_{1} \leq \cdots \leq \omega_{p}$. Then, \eqref{bounded but spike existit} reduces to $\operatorname{tr}\left(\mathbf{W}\boldsymbol{\Sigma}\right)^4=o\left[\operatorname{tr}^{2}\left(\mathbf{W} \boldsymbol{\Sigma}\right)^{2}\right]$, which holds when $\alpha^2=O\left(p^{-3/4}\right)$, $\operatorname{tr}\boldsymbol{\Sigma}^4=o\left(\operatorname{tr}^2\boldsymbol{\Sigma}^2\right)$ and $p^{1/2}\lambda_p^{*2}=o\left(\operatorname{tr}\boldsymbol{\Sigma}^2\right)$. Assumption D is intended to establish the asymptotic normality of our suggested statistic under the local alternative hypothesis and evaluate the power performance of our proposed test, which will hold automatically when the null hypothesis is true. Furthermore, under Assumption D, it can be easily proved that $\operatorname{Var}\left(T_n\right) = \sigma_{n,q_{1}}^{2} (1+o(1))$, i.e., the variance of $T_n$ is dominated by $\sigma_{n,q_1}^2$. Therefore, we use $\sigma_{n,q}^2\triangleq\sigma_{n,q_1}^2$ to denote the leading term of $\operatorname{Var}\left(T_n\right)$.

The theorem presented below establishes the asymptotic normality of $T_{n}$.
\begin{theorem}\label{mean linear test via the weighted-norm: main result}
	Under Assumptions A-D, as $p$, $n \to \infty$, we have
	\begin{align}
	\frac{T_{n } - \left( \sum_{i=1}^{q} \beta_{i} \boldsymbol{\mu}_{i} \right)^{\top} \mathbf{W} \left( \sum_{i=1}^{q} \beta_{i} \boldsymbol{\mu}_{i} \right) }{\sigma_{n,q}}	\stackrel{\text { d }}{\longrightarrow} \mathrm{N} (0,1),
	\end{align}
	where  $\stackrel{\text { d }} {\longrightarrow}$ means the convergence in distribution.
\end{theorem}
%

The proof of Theorem \ref{mean linear test via the weighted-norm: main result} is outlined in Appendix \ref{Appendix2}.
In particular, when the null hypothesis holds, we have the following corollary:
\begin{corollary}
	Under Assumptions A-D and $ \mathrm{H}_0:\sum_{i=1}^{q} \beta_{i} \boldsymbol{\mu}_{i}=\mathbf{0}$, as $p$, $n \to \infty$, we have
	\begin{align}
	\frac{T_{n } }{\sigma_{n,q}}	\stackrel{\text { d }}{\longrightarrow} \mathrm{N} (0,1),
	\end{align}	
\end{corollary}
As the quantities $\operatorname{tr}\left(  \mathbf{W} \boldsymbol{\Sigma}_{i_1}  \mathbf{W} \boldsymbol{\Sigma}_{i_2}\right)$ and $\operatorname{tr}\left( \mathbf{W} \boldsymbol{\Sigma}_{i}\right)^2$ are unknown, it becomes necessary to employ their estimators that enjoy ratio consistency to formulate our test procedure based on $T_{n}$. 
The estimators proposed by \cite{li2024test} are utilized for this purpose, as they offer both consistency and computational efficiency, enabling an effective implementation of our testing approach. They are defined as
\begin{align*}
&\operatorname{tr} \widehat{ \left(\mathbf{W} \boldsymbol{\Sigma}_{i}\right)^{2} } \\
=& \frac{- 1}{(n_{i} - 1) (n_{i} - 2) } \sum_{j=1 }^{n_{i}} \left( ( \mathbf{x}_{ij} -  \bar{\mathbf{x}}_{i} )^{\top} \mathbf{W} ( \mathbf{x}_{ij} -  \bar{\mathbf{x}}_{i} )\right)^{2} + \frac{ (n_{i} - 1)^{2}}{n_{i}  (n_{i} - 3) }\operatorname{tr} (\mathbf{W} \mathbf{S}_{i})^{2} \\
&+ \frac{(n_{i} - 1)}{n_{i}(n_{i} - 2) (n_{i} - 3) } \operatorname{tr}^{2} (\mathbf{W} \mathbf{S}_{i})
\end{align*}
and
\begin{align*}
\operatorname{tr} \widehat{\left(\mathbf{W} \boldsymbol{\Sigma}_{i_{1}} \mathbf{W} \boldsymbol{\Sigma}_{i_{2}}\right)} 
= & \operatorname{tr} (\mathbf{W} \mathbf{S}_{i_{1}} \mathbf{W} \mathbf{S}_{i_{2}}  ),
\end{align*}
where $\mathbf{\bar{x}}_{i}$ and $\mathbf{S}_{i}$ stand for the sample mean vector and the sample covariance matrix of the $i$th group, respectively, i.e.,
\begin{align*}
\mathbf{\bar{x}}_{i}=\frac1{n_i}\sum_{j=1}^{n_i}\mathbf{x}_{ij},\quad
\mathbf{S}_{i} = \frac{1}{n_{i}-1} \sum_{j=1}^{n_{i}} (\mathbf{x}_{ij} - \bar{\mathbf{x}}_{i})(\mathbf{x}_{ij} - \bar{\mathbf{x}}_{i})^{\top}.
\end{align*}
Hence, a ratio consistent estimator of $\sigma_{n,q}^{2}$ can be taken as
\begin{align*}
\hat{\sigma}_{n, q}^{2}=  2 \sum_{i_{1} \neq i_{2}} \frac{\beta_{i_{1}} ^{2}\beta_{i_{2}}^{2}}{n_{i_{1}} n_{i_{2}}} \operatorname{tr} \widehat{\left(\mathbf{W} \boldsymbol{\Sigma}_{i_{1}} \mathbf{W} \boldsymbol{\Sigma}_{i_{2}}\right)}+2 \sum_{i=1}^{q} \frac{\beta_{i} ^{4}}{n_{i}\left(n_{i}-1\right)} \operatorname{tr} \widehat{ \left(\mathbf{W} \boldsymbol{\Sigma}_{i}\right)^{2} }.
\end{align*}
Thus by Slutsky's theorem, we obtain the following theorem:
\begin{theorem}\label{Corollary of CLT of mean linear test via the weighted-norm}
	Under Assumptions A-D, and $ \mathrm{H}_0:\sum_{i=1}^{q} \beta_{i} \boldsymbol{\mu}_{i}=\mathbf{0}$ as $p$, $n \to \infty$, we have
	\begin{align}
	\frac{T_{n } }{\hat{\sigma}_{n, q}}	\stackrel{\text { d }}{\longrightarrow} \mathrm{N} (0,1),
	\end{align}	
\end{theorem}
According to Theorem \ref{Corollary of CLT of mean linear test via the weighted-norm}, our proposed test rejects $\mathrm{H}_{0}$ if $T_{n } \geq \hat{\sigma}_{n, q} z_{\vartheta}$ for a given significance level $\vartheta$, where  $z_{\vartheta}$  is the upper-$\vartheta$  quantile of $N(0,1)$. 


\subsection{The power analysis}\label{section2.3}
In this subsection, we delve into the power of the proposed test and undertake a preliminary analysis of our suggested statistic under the local alternative hypothesis. From Theorem \ref{mean linear test via the weighted-norm: main result}, we can readily deduce the power of our test as follows: 
\begin{align}\label{general power}
\lim _{ n, p \rightarrow \infty} \mathrm{P}\left( T_{n } \geq \hat{\sigma}_{n, q} z_{\vartheta} \right) 
= \lim _{ n, p \rightarrow \infty} \Phi\left\{-z_{\vartheta}+\frac{ \left( \sum_{i=1}^{q} \beta_{i} \boldsymbol{\mu}_{i} \right)^{\top} \mathbf{W} \left( \sum_{i=1}^{q} \beta_{i} \boldsymbol{\mu}_{i} \right) }{ 	{\sigma}_{n, q}  }\right\},
\end{align}
where $ \Phi(\cdot)$ is the cumulative distribution function of the standard normal random variable.
In order to simplify our analysis and underscore the strengths of our method, we assume that the $q$ population covariance matrices are equal. In this scenario, the asymptotic power \eqref{general power} will be transformed to  
\begin{align}\notag
& \lim _{ n, p \rightarrow \infty} \mathrm{P}\left( T_{n } \geq \hat{\sigma}_{n, q_{1}} z_{\vartheta} \right) \\
\label{equal power}
= & \lim _{ n, p \rightarrow \infty} \Phi\left\{-z_{\vartheta}+\frac{ \left( \sum_{i=1}^{q} \beta_{i} \boldsymbol{\mu}_{i} \right)^{\top} \mathbf{W} \left( \sum_{i=1}^{q} \beta_{i} \boldsymbol{\mu}_{i} \right) }{\sqrt{2 \operatorname{tr}\left( \mathbf{W} \boldsymbol{\Sigma}_1 \right)^2  }  \sqrt{ \sum_{i_{1} \neq i_{2}} \frac{\beta_{i_{1}}^{2}\beta_{i_{2}} ^{2}}{n_{i_{1}} n_{i_{2}}} +  \sum_{i=1}^{q} \frac{\beta_{i} ^{4}}{n_{i}\left(n_{i}-1\right)}  } }\right\}.
\end{align}
With an understanding of the weight matrix $\mathbf{W}$, our next step involves deriving the lower bound for the expression in \eqref{equal power} by considering the case where  $\alpha_{1}=\cdots=\alpha_{p}=\alpha$  and  $0<\omega_{1} \leq   \cdots \leq \omega_{p}$. Let  $\lambda_{1} \leq \lambda_{2} \cdots \leq   \lambda_{p}$  and  $\lambda_{1}^{*} \leq \lambda_{2}^{*} \cdots \leq \lambda_{p}^{*}$  be eigenvalues of  $\mathbf{W}$  and  $\boldsymbol{\Sigma}_{1} = \boldsymbol{\Sigma}_{2}= \cdots=\boldsymbol{\Sigma}_{q}$, respectively. Then by calculating the characteristic function of $\mathbf{W}$, it can be found that  $\omega_1^2\leq\lambda_1\leq\omega_2^2,\dots,\omega_{p-1}^2\leq\lambda_{p-1}\leq\omega_p^2$, and $\left|\lambda_p-\omega_p^2\right|\leq\left\|\mathbf{W}-\mathbf{\Omega}\right\|=p\alpha^2 $. For the true mean vectors $\boldsymbol{\mu}_i$, we consider a scenario with weakly dense, yet non-trivial signals across the linear combination of these vectors:
\begin{align*}
\sum_{i=1}^q \beta_i \boldsymbol{\mu}_i \triangleq \boldsymbol{\mu} =(\overbrace{\nu, \ldots, \nu}^{p^{\delta}}, \overbrace{0, \ldots, 0}^{ p^{1-\delta}})^{\top},
\end{align*}
that is, $\boldsymbol{\mu}$ had $p^{\delta}$ nonzero entries of equal value. By noting that $\operatorname{tr}\left( \mathbf{W} \boldsymbol{\Sigma}_1 \right)^2\leq\lambda_p^{*2}\operatorname{tr}\mathbf{W}^2$, 
we can establish a lower bound such that
\begin{align*}
&\frac{ \left( \sum_{i=1}^{q} \beta_{i} \boldsymbol{\mu}_{i} \right)^{\top} \mathbf{W} \left( \sum_{i=1}^{q} \beta_{i} \boldsymbol{\mu}_{i} \right) }{\sqrt{2 \operatorname{tr}\left( \mathbf{W} \boldsymbol{\Sigma}_1 \right)^2  }  \sqrt{ \sum_{i_{1} \neq i_{2}} \frac{\beta_{i_{1}}^{2}\beta_{i_{2}}^{2}}{n_{i_{1}} n_{i_{2}}} +  \sum_{i=1}^{q} \frac{\beta_{i} ^{4}}{n_{i}\left(n_{i}-1\right)}  } } \\
\geq& \frac{\alpha^{2} p^{2 \delta} \nu^{2}+\nu^{2} \sum_{i=1}^{p^{\delta} } \omega_{i}^{2}}{\sqrt{2\left(\lambda_{p}^{*}\right)^{2}\left(\sum_{i=2}^{p} \omega_{i}^{4} + \omega_{p}^{4} +2p \omega_{p}^{2} \alpha^{2} + p^{2} \alpha^{4} \right)}   \sqrt{ \sum_{i_{1} \neq i_{2}} \frac{\beta_{i_{1}} ^{2}\beta_{i_{2}} ^{2}}{n_{i_{1}} n_{i_{2}}} +  \sum_{i=1}^{q} \frac{\beta_{i}^{4}}{n_{i}\left(n_{i}-1\right)}  }  } .
\end{align*}
Then, we have the following corollary:
\begin{corollary}\label{powerto1}
	Assume $\boldsymbol{\Sigma}_{1} = \boldsymbol{\Sigma}_{2}= \cdots=\boldsymbol{\Sigma}_{q}$, $\alpha_{1}=\cdots=\alpha_{p}=\alpha$  and  $0 < \omega_{1} \leq   \cdots \leq \omega_{p} < \infty$, $\alpha^{2}=O\left(p^{-3 / 4}\right)$, $\lambda_{p}^{*} \sqrt{ \sum_{i_{1} \neq i_{2}} \frac{\beta_{i_{1}}^{2}\beta_{i_{2}}^{2}}{n_{i_{1}} n_{i_{2}}} +  \sum_{i=1}^{q} \frac{\beta_{i}^{4}}{n_{i}\left(n_{i}-1\right)}  } =o\left( \nu^{2} p^{2 \delta -1}\right)$ and Assumptions A-D hold. Then under $ \mathrm{H}_{1}:  \sum_{i=1}^q \beta_i \boldsymbol{\mu}_i \neq 0 $, as $p$, $n \to \infty$, the asymptotic power of our proposed test is given by
	\begin{align*}
	\lim _{ n, p \rightarrow \infty} \mathrm{P}\left( T_{n } \geq \hat{\sigma}_{n, q_{1}} z_{\vartheta} \right) 
	= 1.
	\end{align*}
\end{corollary}
Corollary \ref{powerto1} indicates that for the asymptotic power of $T_n$ to reach 1, a sufficient condition is $\delta>1/2$, provided that the eigenvalues of $\boldsymbol{\Sigma}_{1}$ are restricted away from $0$ and $\nu^{2} = O \left( \sqrt{ \sum_{i_{1} \neq i_{2}} \frac{\beta_{i_{1}}^{2}\beta_{i_{2}} ^{2}}{n_{i_{1}} n_{i_{2}}} +  \sum_{i=1}^{q} \frac{\beta_{i} ^{4}}{n_{i}\left(n_{i}-1\right)} } \right)$. This theoretically demonstrates that our statistic performs well under weakly dense signal situations. From another perspective, the suggested new test, with the exception of the strong sparse non-zero signals, may be unable to offer significant benefits when $\delta <1/2$. 


\section{Simulation Studies}\label{section3}
In this section, we conduct a series of numerical studies to evaluate the performance of our proposed test, denoted by $T_L$. To provide a comprehensive evaluation, we compare it against various testing procedures for linear hypotheses of mean vectors, including those introduced by \cite{zhou2017high} and \cite{zhang2017linear}, which are subsequently referenced as $T_{U}$ and $T_{C}$, respectively. Notably, these two referenced methods represent the U-statistic and chi-square approximation techniques. For simplicity, we set $q=3$ first and assume the samples are generated from
\begin{align}\label{sample}
\mathbf{x}_{i j}=\boldsymbol{\mu}_i +\boldsymbol{\Sigma}_{i}^{1/2} \mathbf{z}_{i j},\quad j= 1, \ldots, n_i,\quad i=1, 2, 3.
\end{align}
Moreover, the linear coefficients $\beta_i$ are set to be $\beta_1=2$, $\beta_2=-2$ and $\beta_3=-1$. The additional settings for our model are detailed as follows:
\begin{itemize}
	\item [(1)] The covariance structures $\boldsymbol{\Sigma}_{i}$ are considered as the following two cases, representing the common and different covariance matrices, respectively:
	\begin{itemize}
		\item [\textbf{Case 1.}] $\boldsymbol{\Sigma}_{1}=\boldsymbol{\Sigma}_{2}=\boldsymbol{\Sigma}_{3}= \left(2\times0.4^{|i-j|}\right)_{i,j}$.
		\item  [\textbf{Case 2.}]
		$\boldsymbol{\Sigma}_{1}=\left(0.5^{|i-j|}I_{(|i-j|\leq1)}\right)_{i,j}$, $\boldsymbol{\Sigma}_{2}=1.5\boldsymbol{\Sigma}_{1}$ and $\boldsymbol{\Sigma}_{3}=2\boldsymbol{\Sigma}_{1}$.
	\end{itemize}
	\item [(2)] The elements of $\mathbf{z}_{i j}$ are assumed to be independent and identically distributed, sourced from the following distributions:
	\begin{itemize}
		\item The standard normal distribution $N(0,1)$.
		\item The standardized $Gamma(4,1)$ distribution.
		\item The standardized $t(5)$ distribution.
	\end{itemize}
	\item [(3)] Under the null hypothesis, the mean vectors $\boldsymbol{\mu}_i$ are chosen to be $\boldsymbol{\mu}_1=\boldsymbol{\mu}_2=\boldsymbol{\mu}_3=\boldsymbol{0}$. Conversely, for the alternative hypothesis, the mean vectors are considered as follows:
	\begin{align*}
	\boldsymbol{\mu}_{1}=(\kappa,\ldots,\kappa)^{\top},\quad
	\boldsymbol{\mu}_{2}= (\overbrace{0 , \ldots, 0}^{ [p^{1-\rho}] }, \kappa, \ldots, \kappa )^{\top},\quad
	\boldsymbol{\mu}_{3}= (\overbrace{\kappa , \ldots, \kappa }^{ [p^{1-\rho}] }, 0, \ldots, 0 )^{\top},
	\end{align*}
	where $[a]$ denotes the integer part of $a$. Such designs yield that $\sum_{i=1}^3\beta_i\boldsymbol{\mu}_{i}=(\overbrace{\kappa , \ldots, \kappa }^{ [p^{1-\rho}] }, 0, \ldots, 0 )^{\top}$. Here, $\kappa$ is selected to be $\sqrt{3 r\log p(n_1^{-1}+n_2^{-1}+n_3^{-1}) }$, with $r=0.04,0.08,0.12$ to delineate the signal from weak to strong, and $\rho$ is chosen to be 0.1,0.2,0.3,0.4 to modulate the signal's sparsity.
	\item [(4)] The dimensional settings are established as $n_1=0.5n^*$, $n_2=n^*$, $n_3=1.5n^*$ for $n^*=80,120$, and $p=200,400,600$.
	\item [(5)] For the tuning parameter of our test, we adopt 
	$\alpha_{1}=\cdots=\alpha_{p}=\sqrt{5}p^{-3/8}$  and  $\omega_{k}  = \sqrt{2} (1 + \frac{2k}{3p})$ for $k=1, \ldots , p$.
\end{itemize} 

Empirical sizes and powers at 0.05 significance level from 5000 independent replications are reported in Table \ref{size1}-\ref{size2} and \ref{power1normal}-\ref{power2t}, respectively. Results in Table \ref{size1}-\ref{size2} reveal that the test sizes are well-controlled around the 0.05 significance level across all methods, with disparities decreasing as $n$ and $p$ increase. For empirical powers, as depicted in Table \ref{power1normal}-\ref{power2t}, larger dimensions $p$ and signal strength $r$ exhibit higher powers for all the test methods, while all powers decrease with an increasing sparsity $\rho$. These findings confirm Corollary \ref{powerto1}. However, our proposed test $T_L$ showcases superior powers than the other two competitors, especially for lower $\rho$ values, affirming its exceptional performance in detecting weakly dense linear combination signals. 

\begin{table}[htpb]
	\centering
	\caption{Empirical sizes for Case 1.} 
	\label{size1}
	\begin{tabular}{ccccccccccc}
			\toprule
			\multicolumn{1}{c}{\multirow{2}{*}{$p$}} & \multicolumn{1}{c}{\multirow{2}{*}{$n^*$}} &  \multicolumn{3}{c}{$N(0,1)$} &  \multicolumn{3}{c}{$Gamma(4,1)$} & \multicolumn{3}{c}{$t(5)$}  \\
			\cmidrule(r){3-5} \cmidrule(r){6-8} \cmidrule(r){9-11} 
			&&$T_{L}$&$T_{U}$&$T_{C}$&$T_{L}$&$T_{U}$&$T_{C}$&$T_{L}$&$T_{U}$&$T_{C}$ \\
			\hline
			200 & 80  & 0.059 & 0.059 & 0.054 & 0.062 & 0.057 & 0.051 & 0.063 & 0.059 & 0.050 \\
			& 120 & 0.063 & 0.057 & 0.053 & 0.060 & 0.058 & 0.051 & 0.059 & 0.054 & 0.048 \\
			400 & 80  & 0.053 & 0.056 & 0.053 & 0.063 & 0.055 & 0.050 & 0.062 & 0.063 & 0.056 \\
			& 120 & 0.058 & 0.055 & 0.052 & 0.058 & 0.055 & 0.052 & 0.061 & 0.058 & 0.052 \\
			600 & 80  & 0.057 & 0.057 & 0.054 & 0.057 & 0.056 & 0.051 & 0.057 & 0.057 & 0.053 \\
			& 120 & 0.057 & 0.055 & 0.054 & 0.056 & 0.053 & 0.050 & 0.050 & 0.050 & 0.046 \\
			\bottomrule
	\end{tabular}
\end{table}

\begin{table}[htpb]
	\centering
	\caption{Empirical sizes for Case 2.} 
	\label{size2}
	\begin{tabular}{ccccccccccc}
			\toprule
			\multicolumn{1}{c}{\multirow{2}{*}{$p$}} & \multicolumn{1}{c}{\multirow{2}{*}{$n^*$}} &  \multicolumn{3}{c}{$N(0,1)$} &  \multicolumn{3}{c}{$Gamma(4,1)$} & \multicolumn{3}{c}{$t(5)$}  \\
			\cmidrule(r){3-5} \cmidrule(r){6-8} \cmidrule(r){9-11} 
			&&$T_{L}$&$T_{U}$&$T_{C}$&$T_{L}$&$T_{U}$&$T_{C}$&$T_{L}$&$T_{U}$&$T_{C}$ \\
			\hline
			200 & 80  & 0.059 & 0.056 & 0.052 & 0.063 & 0.054 & 0.050 & 0.064 & 0.053 & 0.047 \\
			& 120 & 0.056 & 0.057 & 0.050 & 0.061 & 0.057 & 0.050 & 0.062 & 0.057 & 0.048 \\
			400 & 80  & 0.063 & 0.063 & 0.057 & 0.058 & 0.054 & 0.050 & 0.059 & 0.055 & 0.048 \\
			& 120 & 0.058 & 0.058 & 0.052 & 0.058 & 0.056 & 0.052 & 0.059 & 0.056 & 0.051 \\
			600 & 80  & 0.056 & 0.055 & 0.052 & 0.055 & 0.055 & 0.052 & 0.058 & 0.056 & 0.049 \\
			& 120 & 0.049 & 0.051 & 0.048 & 0.054 & 0.054 & 0.051 & 0.056 & 0.050 & 0.045 \\
			\bottomrule
	\end{tabular}
\end{table}

\begin{table*}[htpb]
	\centering
	\caption{Empirical powers for Case 1 with $N(0,1)$ distributed $\mathbf{z}_{i j}$.} 
	\label{power1normal}
	\scalebox{0.7}
	{
		\begin{tabular}{ccccccccccccccc}
			\toprule
			\multirow{2}{*}{$p$} & \multirow{2}{*}{$n^*$} & \multirow{2}{*}{$r$} & 
			\multicolumn{3}{c}{$\rho=0.1$} & \multicolumn{3}{c}{$\rho=0.2$} & \multicolumn{3}{c}{$\rho=0.3$} & \multicolumn{3}{c}{$\rho=0.4$} \\ \cmidrule(r){4-6} \cmidrule(r){7-9} \cmidrule(r){10-12} \cmidrule(r){13-15}
			&&&$T_{L}$&$T_{U}$&$T_{C}$    &$T_{L}$&$T_{U}$&$T_{C}$ &$T_{L}$&$T_{U}$&$T_{C}$ &$T_{L}$&$T_{U}$&$T_{C}$  \\ \hline
			200                & 80                 & 0.04               & 0.348  & 0.137  & 0.128 & 0.166  & 0.097  & 0.090 & 0.104  & 0.088  & 0.080 & 0.081  & 0.078  & 0.071 \\
			&                    & 0.08               & 0.603  & 0.251  & 0.239 & 0.277  & 0.149  & 0.140 & 0.131  & 0.110  & 0.103 & 0.090  & 0.088  & 0.082 \\
			&                    & 0.12               & 0.786  & 0.377  & 0.358 & 0.384  & 0.221  & 0.210 & 0.180  & 0.144  & 0.137 & 0.096  & 0.101  & 0.092 \\
			400                & 80                 & 0.04               & 0.554  & 0.184  & 0.173 & 0.220  & 0.114  & 0.108 & 0.109  & 0.092  & 0.087 & 0.070  & 0.068  & 0.065 \\
			&                    & 0.08               & 0.877  & 0.399  & 0.391 & 0.395  & 0.208  & 0.202 & 0.156  & 0.127  & 0.119 & 0.091  & 0.082  & 0.077 \\
			&                    & 0.12               & 0.966  & 0.605  & 0.595 & 0.569  & 0.302  & 0.292 & 0.229  & 0.167  & 0.160 & 0.115  & 0.110  & 0.106 \\
			600                & 80                 & 0.04               & 0.727  & 0.231  & 0.227 & 0.270  & 0.131  & 0.126 & 0.121  & 0.088  & 0.084 & 0.075  & 0.070  & 0.066 \\
			&                    & 0.08               & 0.957  & 0.522  & 0.512 & 0.500  & 0.255  & 0.246 & 0.184  & 0.130  & 0.125 & 0.093  & 0.089  & 0.085 \\
			&                    & 0.12               & 0.997  & 0.770  & 0.763 & 0.696  & 0.388  & 0.379 & 0.260  & 0.184  & 0.179 & 0.118  & 0.113  & 0.108 \\
			200                & 120                & 0.04               & 0.355  & 0.148  & 0.134 & 0.154  & 0.098  & 0.091 & 0.092  & 0.083  & 0.075 & 0.074  & 0.066  & 0.061 \\
			&                    & 0.08               & 0.614  & 0.255  & 0.239 & 0.267  & 0.162  & 0.149 & 0.132  & 0.102  & 0.095 & 0.085  & 0.087  & 0.081 \\
			&                    & 0.12               & 0.791  & 0.396  & 0.379 & 0.376  & 0.219  & 0.207 & 0.177  & 0.149  & 0.141 & 0.100  & 0.106  & 0.098 \\
			400                & 120                & 0.04               & 0.560  & 0.190  & 0.183 & 0.217  & 0.114  & 0.107 & 0.110  & 0.091  & 0.086 & 0.075  & 0.069  & 0.064 \\
			&                    & 0.08               & 0.868  & 0.392  & 0.381 & 0.394  & 0.203  & 0.196 & 0.160  & 0.120  & 0.113 & 0.093  & 0.086  & 0.081 \\
			&                    & 0.12               & 0.974  & 0.611  & 0.600 & 0.567  & 0.313  & 0.301 & 0.210  & 0.169  & 0.161 & 0.114  & 0.107  & 0.100 \\
			600                & 120                & 0.04               & 0.719  & 0.225  & 0.214 & 0.264  & 0.133  & 0.127 & 0.113  & 0.088  & 0.085 & 0.078  & 0.073  & 0.070 \\
			&                    & 0.08               & 0.961  & 0.495  & 0.484 & 0.484  & 0.247  & 0.237 & 0.179  & 0.139  & 0.133 & 0.092  & 0.092  & 0.089 \\
			&                    & 0.12               & 0.997  & 0.768  & 0.763 & 0.695  & 0.398  & 0.389 & 0.264  & 0.188  & 0.181 & 0.123  & 0.113  & 0.108 \\ 
			\bottomrule
	\end{tabular}}
\end{table*}

\begin{table*}[htpb]
	\centering
	\caption{Empirical powers for Case 1 with $Gamma(4,1)$ distributed $\mathbf{z}_{i j}$.} 
	\label{power1gamma}
	\scalebox{0.7}
	{\begin{tabular}{ccccccccccccccc}
			\toprule
			\multirow{2}{*}{$p$} & \multirow{2}{*}{$n^*$} & \multirow{2}{*}{$r$} & 
			\multicolumn{3}{c}{$\rho=0.1$} & \multicolumn{3}{c}{$\rho=0.2$} & \multicolumn{3}{c}{$\rho=0.3$} & \multicolumn{3}{c}{$\rho=0.4$} \\ \cmidrule(r){4-6} \cmidrule(r){7-9} \cmidrule(r){10-12} \cmidrule(r){13-15}
			&&&$T_{L}$&$T_{U}$&$T_{C}$    &$T_{L}$&$T_{U}$&$T_{C}$ &$T_{L}$&$T_{U}$&$T_{C}$ &$T_{L}$&$T_{U}$&$T_{C}$  \\ \hline
			200 & 80  & 0.04 & 0.342 & 0.140 & 0.130 & 0.165 & 0.097 & 0.086 & 0.089 & 0.075 & 0.067 & 0.080 & 0.069 & 0.063 \\
			&     & 0.08 & 0.615 & 0.255 & 0.239 & 0.271 & 0.160 & 0.147 & 0.127 & 0.108 & 0.098 & 0.086 & 0.085 & 0.074 \\
			&     & 0.12 & 0.783 & 0.386 & 0.370 & 0.388 & 0.227 & 0.216 & 0.171 & 0.139 & 0.128 & 0.102 & 0.104 & 0.096 \\
			400 & 80  & 0.04 & 0.575 & 0.190 & 0.177 & 0.219 & 0.121 & 0.113 & 0.107 & 0.086 & 0.080 & 0.070 & 0.071 & 0.067 \\
			&     & 0.08 & 0.876 & 0.396 & 0.380 & 0.398 & 0.210 & 0.199 & 0.162 & 0.123 & 0.114 & 0.103 & 0.089 & 0.084 \\
			&     & 0.12 & 0.972 & 0.601 & 0.586 & 0.563 & 0.306 & 0.293 & 0.230 & 0.174 & 0.163 & 0.103 & 0.106 & 0.100 \\
			600 & 80  & 0.04 & 0.709 & 0.224 & 0.213 & 0.277 & 0.138 & 0.130 & 0.121 & 0.095 & 0.089 & 0.076 & 0.070 & 0.065 \\
			&     & 0.08 & 0.955 & 0.493 & 0.483 & 0.505 & 0.246 & 0.235 & 0.193 & 0.131 & 0.124 & 0.095 & 0.090 & 0.084 \\
			&     & 0.12 & 0.996 & 0.761 & 0.752 & 0.701 & 0.393 & 0.380 & 0.262 & 0.194 & 0.185 & 0.113 & 0.119 & 0.113 \\
			200 & 120 & 0.04 & 0.346 & 0.134 & 0.123 & 0.167 & 0.097 & 0.088 & 0.096 & 0.078 & 0.071 & 0.070 & 0.068 & 0.061 \\
			&     & 0.08 & 0.620 & 0.247 & 0.231 & 0.280 & 0.155 & 0.141 & 0.136 & 0.110 & 0.101 & 0.086 & 0.082 & 0.074 \\
			&     & 0.12 & 0.788 & 0.382 & 0.364 & 0.385 & 0.226 & 0.211 & 0.177 & 0.145 & 0.133 & 0.100 & 0.102 & 0.092 \\
			400 & 120 & 0.04 & 0.562 & 0.187 & 0.176 & 0.226 & 0.121 & 0.114 & 0.104 & 0.086 & 0.080 & 0.072 & 0.071 & 0.065 \\
			&     & 0.08 & 0.857 & 0.387 & 0.372 & 0.400 & 0.200 & 0.188 & 0.156 & 0.119 & 0.110 & 0.092 & 0.092 & 0.085 \\
			&     & 0.12 & 0.977 & 0.616 & 0.600 & 0.559 & 0.307 & 0.295 & 0.237 & 0.173 & 0.165 & 0.116 & 0.105 & 0.098 \\
			600 & 120 & 0.04 & 0.715 & 0.231 & 0.220 & 0.267 & 0.123 & 0.118 & 0.115 & 0.095 & 0.088 & 0.082 & 0.078 & 0.072 \\
			&     & 0.08 & 0.959 & 0.518 & 0.503 & 0.504 & 0.255 & 0.247 & 0.192 & 0.129 & 0.121 & 0.096 & 0.089 & 0.081 \\
			&     & 0.12 & 0.996 & 0.770 & 0.758 & 0.689 & 0.387 & 0.376 & 0.262 & 0.194 & 0.184 & 0.125 & 0.117 & 0.112 \\ 
				\bottomrule
	\end{tabular}}
\end{table*}

\begin{table*}[htpb]
	\centering
	\caption{Empirical powers for Case 1 with $t(5)$ distributed $\mathbf{z}_{i j}$.} 
	\label{power1t}
	\scalebox{0.7}
	{\begin{tabular}{ccccccccccccccc}
			\toprule
			\multirow{2}{*}{$p$} & \multirow{2}{*}{$n^*$} & \multirow{2}{*}{$r$} & 
			\multicolumn{3}{c}{$\rho=0.1$} & \multicolumn{3}{c}{$\rho=0.2$} & \multicolumn{3}{c}{$\rho=0.3$} & \multicolumn{3}{c}{$\rho=0.4$} \\ \cmidrule(r){4-6} \cmidrule(r){7-9} \cmidrule(r){10-12} \cmidrule(r){13-15}
			&&&$T_{L}$&$T_{U}$&$T_{C}$    &$T_{L}$&$T_{U}$&$T_{C}$ &$T_{L}$&$T_{U}$&$T_{C}$ &$T_{L}$&$T_{U}$&$T_{C}$  \\ \hline
			200 & 80  & 0.04 & 0.349 & 0.138 & 0.125 & 0.165 & 0.108 & 0.093 & 0.101 & 0.079 & 0.070 & 0.081 & 0.070 & 0.059 \\
			&     & 0.08 & 0.615 & 0.255 & 0.234 & 0.261 & 0.159 & 0.143 & 0.138 & 0.106 & 0.096 & 0.091 & 0.088 & 0.078 \\
			&     & 0.12 & 0.789 & 0.390 & 0.364 & 0.387 & 0.222 & 0.204 & 0.175 & 0.135 & 0.122 & 0.109 & 0.110 & 0.097 \\
			400 & 80  & 0.04 & 0.569 & 0.190 & 0.172 & 0.217 & 0.114 & 0.101 & 0.106 & 0.087 & 0.077 & 0.070 & 0.073 & 0.065 \\
			&     & 0.08 & 0.864 & 0.395 & 0.372 & 0.398 & 0.203 & 0.186 & 0.166 & 0.123 & 0.109 & 0.088 & 0.090 & 0.080 \\
			&     & 0.12 & 0.968 & 0.612 & 0.586 & 0.565 & 0.309 & 0.289 & 0.231 & 0.173 & 0.158 & 0.118 & 0.114 & 0.102 \\
			600 & 80  & 0.04 & 0.719 & 0.239 & 0.223 & 0.266 & 0.133 & 0.123 & 0.115 & 0.089 & 0.081 & 0.076 & 0.074 & 0.066 \\
			&     & 0.08 & 0.956 & 0.517 & 0.494 & 0.489 & 0.239 & 0.221 & 0.194 & 0.143 & 0.132 & 0.098 & 0.094 & 0.086 \\
			&     & 0.12 & 0.997 & 0.760 & 0.742 & 0.695 & 0.385 & 0.367 & 0.263 & 0.202 & 0.186 & 0.115 & 0.107 & 0.098 \\
			200 & 120 & 0.04 & 0.341 & 0.141 & 0.127 & 0.163 & 0.098 & 0.086 & 0.099 & 0.087 & 0.076 & 0.074 & 0.068 & 0.059 \\
			&     & 0.08 & 0.613 & 0.254 & 0.233 & 0.261 & 0.150 & 0.137 & 0.133 & 0.104 & 0.093 & 0.089 & 0.091 & 0.081 \\
			&     & 0.12 & 0.791 & 0.396 & 0.368 & 0.386 & 0.227 & 0.208 & 0.168 & 0.135 & 0.124 & 0.100 & 0.106 & 0.094 \\
			400 & 120 & 0.04 & 0.549 & 0.179 & 0.167 & 0.214 & 0.121 & 0.109 & 0.107 & 0.078 & 0.071 & 0.070 & 0.066 & 0.059 \\
			&     & 0.08 & 0.872 & 0.385 & 0.364 & 0.403 & 0.204 & 0.189 & 0.158 & 0.123 & 0.113 & 0.090 & 0.087 & 0.079 \\
			&     & 0.12 & 0.967 & 0.614 & 0.594 & 0.568 & 0.317 & 0.296 & 0.234 & 0.181 & 0.167 & 0.112 & 0.111 & 0.104 \\
			600 & 120 & 0.04 & 0.723 & 0.232 & 0.218 & 0.267 & 0.128 & 0.119 & 0.123 & 0.091 & 0.083 & 0.081 & 0.074 & 0.067 \\
			&     & 0.08 & 0.959 & 0.508 & 0.489 & 0.487 & 0.239 & 0.223 & 0.184 & 0.139 & 0.130 & 0.100 & 0.089 & 0.081 \\
			&     & 0.12 & 0.996 & 0.760 & 0.747 & 0.682 & 0.383 & 0.364 & 0.266 & 0.200 & 0.186 & 0.122 & 0.116 & 0.106 \\ \bottomrule
	\end{tabular}}
\end{table*}

\begin{table*}[htpb]
	\centering
	\caption{Empirical powers for Case 2 with $N(0,1)$ distributed $\mathbf{z}_{i j}$.} 
	\label{power2normal}
	\scalebox{0.7}
	{\begin{tabular}{ccccccccccccccc}
			\toprule
			\multirow{2}{*}{$p$} & \multirow{2}{*}{$n^*$} & \multirow{2}{*}{$r$} & 
			\multicolumn{3}{c}{$\rho=0.1$} & \multicolumn{3}{c}{$\rho=0.2$} & \multicolumn{3}{c}{$\rho=0.3$} & \multicolumn{3}{c}{$\rho=0.4$} \\ \cmidrule(r){4-6} \cmidrule(r){7-9} \cmidrule(r){10-12} \cmidrule(r){13-15}
			&&&$T_{L}$&$T_{U}$&$T_{C}$    &$T_{L}$&$T_{U}$&$T_{C}$ &$T_{L}$&$T_{U}$&$T_{C}$ &$T_{L}$&$T_{U}$&$T_{C}$  \\ \hline
			200 & 80  & 0.04 & 0.547 & 0.197 & 0.183 & 0.248 & 0.126 & 0.118 & 0.119 & 0.095 & 0.086 & 0.075 & 0.075 & 0.070 \\
			&     & 0.08 & 0.858 & 0.420 & 0.405 & 0.435 & 0.236 & 0.222 & 0.198 & 0.147 & 0.135 & 0.123 & 0.105 & 0.098 \\
			&     & 0.12 & 0.962 & 0.631 & 0.616 & 0.624 & 0.363 & 0.345 & 0.271 & 0.207 & 0.197 & 0.135 & 0.138 & 0.127 \\
			400 & 80  & 0.04 & 0.814 & 0.303 & 0.292 & 0.346 & 0.163 & 0.155 & 0.146 & 0.107 & 0.102 & 0.087 & 0.080 & 0.074 \\
			&     & 0.08 & 0.987 & 0.647 & 0.635 & 0.636 & 0.337 & 0.327 & 0.258 & 0.186 & 0.178 & 0.116 & 0.105 & 0.099 \\
			&     & 0.12 & 1.000 & 0.889 & 0.882 & 0.844 & 0.541 & 0.528 & 0.365 & 0.259 & 0.249 & 0.146 & 0.138 & 0.130 \\
			600 & 80  & 0.04 & 0.931 & 0.388 & 0.377 & 0.422 & 0.181 & 0.175 & 0.159 & 0.119 & 0.112 & 0.093 & 0.085 & 0.081 \\
			&     & 0.08 & 0.999 & 0.813 & 0.805 & 0.759 & 0.419 & 0.408 & 0.293 & 0.204 & 0.200 & 0.118 & 0.114 & 0.108 \\
			&     & 0.12 & 1.000 & 0.970 & 0.969 & 0.926 & 0.653 & 0.644 & 0.432 & 0.316 & 0.306 & 0.166 & 0.152 & 0.146 \\
			200 & 120 & 0.04 & 0.557 & 0.198 & 0.183 & 0.244 & 0.133 & 0.123 & 0.118 & 0.093 & 0.085 & 0.089 & 0.076 & 0.069 \\
			&     & 0.08 & 0.862 & 0.427 & 0.407 & 0.438 & 0.232 & 0.219 & 0.197 & 0.146 & 0.133 & 0.110 & 0.103 & 0.094 \\
			&     & 0.12 & 0.965 & 0.640 & 0.622 & 0.614 & 0.352 & 0.336 & 0.275 & 0.205 & 0.192 & 0.141 & 0.136 & 0.124 \\
			400 & 120 & 0.04 & 0.824 & 0.301 & 0.291 & 0.332 & 0.162 & 0.153 & 0.144 & 0.111 & 0.104 & 0.090 & 0.083 & 0.078 \\
			&     & 0.08 & 0.988 & 0.657 & 0.643 & 0.638 & 0.335 & 0.318 & 0.245 & 0.172 & 0.162 & 0.116 & 0.107 & 0.101 \\
			&     & 0.12 & 0.999 & 0.886 & 0.879 & 0.834 & 0.532 & 0.520 & 0.357 & 0.262 & 0.249 & 0.157 & 0.154 & 0.146 \\
			600 & 120 & 0.04 & 0.929 & 0.386 & 0.374 & 0.439 & 0.196 & 0.187 & 0.166 & 0.118 & 0.111 & 0.087 & 0.078 & 0.074 \\
			&     & 0.08 & 0.999 & 0.821 & 0.811 & 0.763 & 0.413 & 0.400 & 0.297 & 0.209 & 0.202 & 0.113 & 0.109 & 0.103 \\
			&     & 0.12 & 1.000 & 0.969 & 0.966 & 0.921 & 0.662 & 0.650 & 0.422 & 0.313 & 0.303 & 0.164 & 0.155 & 0.148 \\ \bottomrule
	\end{tabular}}
\end{table*}

\begin{table*}[htpb]
	\centering
	\caption{Empirical powers for Case 2 with $Gamma(4,1)$ distributed $\mathbf{z}_{i j}$.} 
	\label{power2gamma}
	\scalebox{0.7}{
	\begin{tabular}{ccccccccccccccc}
			\toprule
			\multirow{2}{*}{$p$} & \multirow{2}{*}{$n^*$} & \multirow{2}{*}{$r$} & 
			\multicolumn{3}{c}{$\rho=0.1$} & \multicolumn{3}{c}{$\rho=0.2$} & \multicolumn{3}{c}{$\rho=0.3$} & \multicolumn{3}{c}{$\rho=0.4$} \\ \cmidrule(r){4-6} \cmidrule(r){7-9} \cmidrule(r){10-12} \cmidrule(r){13-15}
			&&&$T_{L}$&$T_{U}$&$T_{C}$    &$T_{L}$&$T_{U}$&$T_{C}$ &$T_{L}$&$T_{U}$&$T_{C}$ &$T_{L}$&$T_{U}$&$T_{C}$  \\ \hline
			200 & 80  & 0.04 & 0.538 & 0.200 & 0.184 & 0.237 & 0.131 & 0.119 & 0.124 & 0.089 & 0.079 & 0.083 & 0.076 & 0.067 \\
			&     & 0.08 & 0.859 & 0.425 & 0.403 & 0.443 & 0.234 & 0.217 & 0.196 & 0.145 & 0.131 & 0.105 & 0.107 & 0.098 \\
			&     & 0.12 & 0.969 & 0.648 & 0.626 & 0.625 & 0.371 & 0.348 & 0.276 & 0.212 & 0.195 & 0.140 & 0.135 & 0.123 \\
			400 & 80  & 0.04 & 0.815 & 0.299 & 0.285 & 0.355 & 0.170 & 0.160 & 0.143 & 0.109 & 0.100 & 0.082 & 0.084 & 0.075 \\
			&     & 0.08 & 0.986 & 0.648 & 0.634 & 0.641 & 0.346 & 0.331 & 0.242 & 0.173 & 0.162 & 0.113 & 0.110 & 0.101 \\
			&     & 0.12 & 1.000 & 0.890 & 0.883 & 0.833 & 0.534 & 0.515 & 0.361 & 0.259 & 0.246 & 0.151 & 0.147 & 0.139 \\
			600 & 80  & 0.04 & 0.930 & 0.375 & 0.361 & 0.431 & 0.187 & 0.180 & 0.165 & 0.121 & 0.115 & 0.089 & 0.084 & 0.079 \\
			&     & 0.08 & 0.999 & 0.814 & 0.803 & 0.755 & 0.399 & 0.384 & 0.287 & 0.203 & 0.194 & 0.117 & 0.115 & 0.108 \\
			&     & 0.12 & 1.000 & 0.976 & 0.974 & 0.921 & 0.655 & 0.643 & 0.429 & 0.311 & 0.299 & 0.159 & 0.155 & 0.147 \\
			200 & 120 & 0.04 & 0.559 & 0.205 & 0.189 & 0.227 & 0.122 & 0.107 & 0.117 & 0.089 & 0.079 & 0.087 & 0.077 & 0.070 \\
			&     & 0.08 & 0.851 & 0.416 & 0.392 & 0.440 & 0.238 & 0.222 & 0.183 & 0.146 & 0.133 & 0.102 & 0.110 & 0.100 \\
			&     & 0.12 & 0.964 & 0.647 & 0.627 & 0.627 & 0.361 & 0.339 & 0.264 & 0.204 & 0.186 & 0.140 & 0.130 & 0.120 \\
			400 & 120 & 0.04 & 0.820 & 0.296 & 0.286 & 0.350 & 0.162 & 0.154 & 0.144 & 0.105 & 0.095 & 0.082 & 0.086 & 0.076 \\
			&     & 0.08 & 0.988 & 0.661 & 0.647 & 0.639 & 0.339 & 0.322 & 0.253 & 0.171 & 0.160 & 0.123 & 0.107 & 0.099 \\
			&     & 0.12 & 1.000 & 0.890 & 0.882 & 0.829 & 0.528 & 0.512 & 0.346 & 0.256 & 0.242 & 0.152 & 0.144 & 0.135 \\
			600 & 120 & 0.04 & 0.930 & 0.392 & 0.379 & 0.430 & 0.189 & 0.181 & 0.165 & 0.116 & 0.111 & 0.087 & 0.084 & 0.077 \\
			&     & 0.08 & 0.998 & 0.813 & 0.802 & 0.770 & 0.432 & 0.419 & 0.291 & 0.206 & 0.197 & 0.119 & 0.114 & 0.107 \\
			&     & 0.12 & 1.000 & 0.969 & 0.965 & 0.923 & 0.659 & 0.649 & 0.424 & 0.317 & 0.305 & 0.163 & 0.160 & 0.151 \\ \bottomrule
	\end{tabular}}
\end{table*}

\begin{table*}[htpb]
	\centering
	\caption{Empirical powers for Case 2 with $t(5)$ distributed $\mathbf{z}_{i j}$.} 
	\label{power2t}
	\scalebox{0.7}{
	\begin{tabular}{ccccccccccccccc}
			\toprule
			\multirow{2}{*}{$p$} & \multirow{2}{*}{$n^*$} & \multirow{2}{*}{$r$} & 
			\multicolumn{3}{c}{$\rho=0.1$} & \multicolumn{3}{c}{$\rho=0.2$} & \multicolumn{3}{c}{$\rho=0.3$} & \multicolumn{3}{c}{$\rho=0.4$} \\ \cmidrule(r){4-6} \cmidrule(r){7-9} \cmidrule(r){10-12} \cmidrule(r){13-15}
			&&&$T_{L}$&$T_{U}$&$T_{C}$    &$T_{L}$&$T_{U}$&$T_{C}$ &$T_{L}$&$T_{U}$&$T_{C}$ &$T_{L}$&$T_{U}$&$T_{C}$  \\ \hline
			200 & 80  & 0.04 & 0.569 & 0.201 & 0.181 & 0.241 & 0.129 & 0.115 & 0.125 & 0.096 & 0.085 & 0.075 & 0.071 & 0.063 \\
			&     & 0.08 & 0.858 & 0.421 & 0.395 & 0.434 & 0.223 & 0.203 & 0.200 & 0.145 & 0.130 & 0.105 & 0.104 & 0.092 \\
			&     & 0.12 & 0.970 & 0.638 & 0.613 & 0.620 & 0.367 & 0.343 & 0.261 & 0.203 & 0.182 & 0.135 & 0.129 & 0.116 \\
			400 & 80  & 0.04 & 0.819 & 0.292 & 0.270 & 0.350 & 0.163 & 0.149 & 0.150 & 0.109 & 0.099 & 0.096 & 0.082 & 0.073 \\
			&     & 0.08 & 0.986 & 0.666 & 0.644 & 0.640 & 0.330 & 0.310 & 0.250 & 0.173 & 0.159 & 0.115 & 0.108 & 0.098 \\
			&     & 0.12 & 1.000 & 0.888 & 0.876 & 0.830 & 0.532 & 0.508 & 0.361 & 0.258 & 0.239 & 0.162 & 0.154 & 0.138 \\
			600 & 80  & 0.04 & 0.931 & 0.392 & 0.371 & 0.429 & 0.188 & 0.175 & 0.167 & 0.118 & 0.108 & 0.087 & 0.076 & 0.068 \\
			&     & 0.08 & 1.000 & 0.813 & 0.801 & 0.755 & 0.417 & 0.403 & 0.290 & 0.196 & 0.183 & 0.120 & 0.110 & 0.100 \\
			&     & 0.12 & 1.000 & 0.977 & 0.975 & 0.922 & 0.658 & 0.640 & 0.435 & 0.323 & 0.302 & 0.167 & 0.168 & 0.153 \\
			200 & 120 & 0.04 & 0.564 & 0.206 & 0.185 & 0.230 & 0.125 & 0.112 & 0.121 & 0.094 & 0.084 & 0.082 & 0.073 & 0.065 \\
			&     & 0.08 & 0.865 & 0.423 & 0.396 & 0.450 & 0.238 & 0.220 & 0.197 & 0.145 & 0.129 & 0.108 & 0.099 & 0.091 \\
			&     & 0.12 & 0.968 & 0.644 & 0.622 & 0.626 & 0.370 & 0.347 & 0.267 & 0.207 & 0.189 & 0.143 & 0.134 & 0.118 \\
			400 & 120 & 0.04 & 0.816 & 0.306 & 0.287 & 0.348 & 0.160 & 0.147 & 0.148 & 0.107 & 0.095 & 0.074 & 0.077 & 0.069 \\
			&     & 0.08 & 0.985 & 0.662 & 0.645 & 0.633 & 0.337 & 0.320 & 0.241 & 0.174 & 0.163 & 0.119 & 0.112 & 0.101 \\
			&     & 0.12 & 0.999 & 0.889 & 0.881 & 0.831 & 0.535 & 0.516 & 0.354 & 0.270 & 0.251 & 0.155 & 0.148 & 0.138 \\
			600 & 120 & 0.04 & 0.930 & 0.387 & 0.369 & 0.430 & 0.189 & 0.176 & 0.150 & 0.103 & 0.096 & 0.086 & 0.081 & 0.072 \\
			&     & 0.08 & 0.999 & 0.804 & 0.794 & 0.756 & 0.414 & 0.396 & 0.294 & 0.209 & 0.194 & 0.129 & 0.115 & 0.107 \\
			&     & 0.12 & 1.000 & 0.976 & 0.973 & 0.916 & 0.655 & 0.640 & 0.432 & 0.318 & 0.301 & 0.166 & 0.163 & 0.153 \\ \bottomrule
	\end{tabular}}
\end{table*}

\section{Conclusions and discussions}
In this article, we propose a novel testing procedure for a linear combination of the mean vectors of several populations. 
Some of the existing tests are already included in our suggested framework, such as \cite{chen2010two,li2017test,zhou2017high, jiang022nonparametric}. According to our theoretical analysis and numerical results, the suggested test can have substantial power improvements in the presence of the weakly dense nonzero signal in the linear combination of the true mean vectors. Throughout this paper, we shall limit the case study in this article to the following: $0<\omega_{1} \leq \cdots \leq \omega_{p}$, and $\alpha_{1} = \cdots = \alpha_{p}$. Certainly, it would be interesting to explore other weighting function alternatives. To conclude the article, an intriguing potential future expansion is described below: First, in the case when the sample sizes is fixed but the data dimension is divergent, we might reconsider the original problem \eqref{test on linear hypothesis of means1}, see, e.g., \cite{li2023finite}. Second, it is also interesting to consider the more general linear hypothesis testing problem:
$\mathrm{H}_{0}: \mathbf{G} \mathbf{M}=\mathbf{0}$. where  $\mathbf{M}=\left(\boldsymbol{\mu}_{1}, \ldots, \boldsymbol{\mu}_{q}\right)^{\top}$  is a  $q \times p$  matrix consisting of the  $q$  mean vectors, and  $\mathbf{M}$  is a  $r \times q$ known coefficient matrix. More endeavors in this regard are required.

\backmatter




\bmhead{Funding}
Bai's research was supported by National Natural Science Foundation of China (No. 12271536, No. 12171198), and Natural Science Foundation of Jilin Province (No. 20210101147JC). 
\section*{Declarations}
\bmhead{Conflict of interest} On behalf of all authors, the corresponding author states that there is no conflict of interest.

\begin{appendices}

\section{Proofs of main results}\label{Appendix}
\subsection{Proof of Lemma \ref{linear mean test using random integration}}\label{Appendix1} 
Assume that $\breve{\mathbf{x}}$ is independent copies of  $\mathbf{x}$, we have by Fubini's theorem,
\begin{align*}
\operatorname{RI}_{w}(\mathbf{x})  =&\int \mathbb{E}^{2}(\boldsymbol{\theta }^{\top} \mathbf{x}) w(\boldsymbol{\theta }) d \boldsymbol{\theta } \\
=&\int \mathbb{E}(\boldsymbol{\theta }^{\top}\mathbf{x}) \mathbb{E}(\boldsymbol{\theta }^{\top}\breve{\mathbf{x}}) w(\boldsymbol{\theta }) d \boldsymbol{\theta } \\
=& \int \mathbb{E}\left(\boldsymbol{\theta }^{\top}\mathbf{x} \boldsymbol{\theta }^{\top}\breve{\mathbf{x}}\right) w(\boldsymbol{\theta }) d \boldsymbol{\theta } \\
=& \mathbb{E}\left\{\int\left(\boldsymbol{\theta }^{\top}\mathbf{x} \boldsymbol{\theta }^{\top}\breve{\mathbf{x}}\right) w(\boldsymbol{\theta }) d \boldsymbol{\theta }\right\} \\
=&\boldsymbol{\mu}^{\top} \mathbf{\Omega} \boldsymbol{\mu} + \left(\boldsymbol{\mu}^{\top} \boldsymbol{\alpha}\right)^{2}.
\end{align*}
\subsection{Proof of Theorem \ref{mean linear test via the weighted-norm: main result}}\label{Appendix2}  
Write
\begin{align}
T_{n}  = T_{n1}  + T_{n2}
\end{align}
where 
\begin{align*}
T_{n1} = \sum_{i_{1} \neq i_{2}} \frac{\beta_{i_{1}} \beta_{i_{2}} }
{n_{i_{1}} n_{i_{2}}} \sum_{j_{1}, j_{2}}\left(\mathbf{x}_{i_{1} j_{1}}\right)^{\top} \mathbf{W} \left(\mathbf{x}_{i_{2} j_{2}} \right)
\end{align*}
and
\begin{align*}
T_{n2} = \sum_{i=1}^{q} \frac{\beta_{i}^{2}}{n_{i}\left(n_{i}-1\right)} \sum_{j_{1} \neq j_{2}}\left(\mathbf{x}_{i j_{1}}\right)^{\top} \mathbf{W} \left(\mathbf{x}_{i j_{2}} \right).
\end{align*}
It is very easy to show that
\begin{align*}
\mathbb{E} \left( T_{n} \right) =& \mathbb{E} \left(  T_{n1} \right) +  \mathbb{E} \left( T_{n2} \right) \\
=&  \sum_{i_{1} \neq i_{2}} \frac{\beta_{i_{1}} \beta_{i_{2}} }
{n_{i_{1}} n_{i_{2}}} \sum_{j_{1}, j_{2}}\left( \mathbb{E} \mathbf{x}_{i_{1} j_{1}}\right)^{\top} \mathbf{W} \left( \mathbb{E} \mathbf{x}_{i_{2} j_{2}} \right) \\
& + \sum_{i=1}^{q} \frac{\beta_{i}^{2}}{n_{i}\left(n_{i}-1\right)} \sum_{j_{1} \neq j_{2}}\left( \mathbb{E} \mathbf{x}_{i j_{1}}\right)^{\top} \mathbf{W} \left( \mathbb{E} \mathbf{x}_{i j_{2}} \right) \\
=&  \sum_{i_{1} \neq i_{2}} \frac{\beta_{i_{1}} \beta_{i_{2}} }
{n_{i_{1}} n_{i_{2}}} \sum_{j_{1}, j_{2}} \boldsymbol{\mu}_{i_{1}}^{\top} \mathbf{W} \boldsymbol{\mu}_{i_{2}} + \sum_{i=1}^{q} \frac{\beta_{i}^{2}}{n_{i}\left(n_{i}-1\right)} \sum_{j_{1} \neq j_{2}}  \boldsymbol{\mu}_{i}^{\top} \mathbf{W} \boldsymbol{\mu}_{i} \\
= & \left( \sum_{i=1}^{q} \beta_{i} \boldsymbol{\mu}_{i} \right)^{\top} \mathbf{W} \left( \sum_{i=1}^{q} \beta_{i} \boldsymbol{\mu}_{i} \right).
\end{align*} 
For the variance of $T_n$, we decompose it into
\begin{align*}
\operatorname{Var} \left( T_{n} \right) =& \operatorname{Var} \left(  T_{n1} \right) +  \operatorname{Var} \left( T_{n2} \right)  + 2\operatorname{Cov} \left( T_{n1}, T_{n2} \right),
\end{align*}
and calculate each term. We have
\begin{align}
&\operatorname{Var} \left(  T_{n1} \right) \notag\\
=& \mathbb{E}\left( \sum_{i_1 \neq i_2} \sum_{i_3 \neq i_4} \frac{\beta_{i_1} \beta_{i_2} \beta_{i_3} \beta_{i_4} }{n_{i_1} n_{i_2} n_{i_3} n_{i_4}} \sum_{j_1, j_2} \sum_{j_3, j_4} \left( \mathbf{x}_{i_1 j_1}\right)^{\top} \mathbf{W} \mathbf{x}_{i_2 j_2}\left(\mathbf{x}_{i_3 j_3}\right)^{\top} \mathbf{W} \mathbf{x}_{i_4 j_4} \right) \notag \\
&-\left( \sum_{i_{1} \neq i_{2}} \frac{\beta_{i_{1}} \beta_{i_{2}} }
{n_{i_{1}} n_{i_{2}}} \sum_{j_{1}, j_{2}} \boldsymbol{\mu}_{i_{1}}^{\top} \mathbf{W} \boldsymbol{\mu}_{i_{2}}  \right)^{2} \notag \\
= &  \sum_{i_1 \neq i_2 \neq i_3 \neq i_4} \beta_{i_1} \beta_{i_2} \beta_{i_3} \beta_{i_4} \boldsymbol{\mu}_{i_1}^{\top} \mathbf{W} \boldsymbol{\mu}_{i_2}  \mathbf{W} \boldsymbol{\mu}_{i_3}^{\top} \boldsymbol{\mu}_{i_4}  + 4 \sum_{i_1 \neq i_2 \neq i_{3}} \frac{  \beta_{i_1} \beta_{i_2} \beta_{i_3} \beta_{i_2} }{ n_{i_1} n_{i_2}^{2} n_{i_3}} \sum_{j_1, j_2} \sum_{j_3, j_4} \mathbb{E}  \left( \left( \mathbf{x}_{i_1 j_1}\right)^{\top}  \notag \right. \\ 
& \left. \cdot \mathbf{W} \mathbf{x}_{i_2 j_2} \left(\mathbf{x}_{i_3 j_3}\right)^{\top} \mathbf{W} \mathbf{x}_{i_2 j_4}  \right) + 2 \sum_{i_1 \neq i_2}   \frac{ \beta_{i_1}^{2} \beta_{i_2}^{2} }{n_{i_1}^{2} n_{i_2}^{2}} \sum_{j_1, j_2} \sum_{j_3, j_4} \mathbb{E} \left(\left( \mathbf{x}_{i_1 j_1}\right)^{\top} \mathbf{W} \mathbf{x}_{i_2 j_2} \left( \mathbf{x}_{i_1 j_3}\right)^{\top} \mathbf{W} \mathbf{x}_{i_2 j_4}  \right) \notag \\ 
&-\left( \sum_{i_{1} \neq i_{2}} \frac{\beta_{i_{1}} \beta_{i_{2}} }
{n_{i_{1}} n_{i_{2}}} \sum_{j_{1}, j_{2}} \boldsymbol{\mu}_{i_{1}}^{\top} \mathbf{W} \boldsymbol{\mu}_{i_{2}}  \right)^{2} \notag  \\
= &  \sum_{i_1 \neq i_2 \neq i_3 \neq i_4} \beta_{i_1} \beta_{i_2} \beta_{i_3} \beta_{i_4} \boldsymbol{\mu}_{i_1}^{\top} \mathbf{W} \boldsymbol{\mu}_{i_2}  \mathbf{W} \boldsymbol{\mu}_{i_3}^{\top} \boldsymbol{\mu}_{i_4}  + 4 \sum_{i_1 \neq i_2 \neq i_{3}} \frac{  \beta_{i_1} \beta_{i_2} \beta_{i_3} \beta_{i_2} }{ n_{i_1} n_{i_2}^{2} n_{i_3}} \sum_{j_1, j_3} \left( \sum_{j_2 \neq j_4} +\sum_{j_2 = j_4}\right)  \notag\\ 
&\mathbb{E}  \left( \left( \mathbf{x}_{i_2 j_2}\right)^{\top}   \mathbf{W} \mathbf{x}_{i_1 j_1} \left(\mathbf{x}_{i_3 j_3}\right)^{\top} \mathbf{W} \mathbf{x}_{i_2 j_4}  \right) + 2 \sum_{i_1 \neq i_2}   \frac{ \beta_{i_1}^{2} \beta_{i_2}^{2} }{n_{i_1}^{2} n_{i_2}^{2}} \sum_{j_2, j_4} \sum_{j_1, j_3} \mathbb{E} \left(\left( \mathbf{x}_{i_1 j_1}\right)^{\top} \mathbf{W} \mathbf{x}_{i_2 j_2} \left( \mathbf{x}_{i_2 j_4}\right)^{\top} \mathbf{W} \mathbf{x}_{i_1 j_3}  \right) \notag \\
&-\left( \sum_{i_{1} \neq i_{2}} \frac{\beta_{i_{1}} \beta_{i_{2}} }
{n_{i_{1}} n_{i_{2}}} \sum_{j_{1}, j_{2}} \boldsymbol{\mu}_{i_{1}}^{\top} \mathbf{W} \boldsymbol{\mu}_{i_{2}}  \right)^{2}\notag  \\
= &  \sum_{i_1 \neq i_2 \neq i_3 \neq i_4} \beta_{i_1} \beta_{i_2} \beta_{i_3} \beta_{i_4} \boldsymbol{\mu}_{i_1}^{\top} \mathbf{W} \boldsymbol{\mu}_{i_2} \mathbf{W} \boldsymbol{\mu}_{i_3}^{\top} \boldsymbol{\mu}_{i_4}  + 4 \sum_{i_1 \neq i_2 \neq i_{3}} \beta_{i_1} \beta_{i_2} \beta_{i_3} \beta_{i_2}   \left( \boldsymbol{\mu}_{i_2}^{\top} \mathbf{W} \boldsymbol{\mu}_{i_1} \boldsymbol{\mu}_{i_3}^{\top} \mathbf{W} \boldsymbol{\mu}_{i_2} + \notag \right. \\
& \left.+ \frac{1}{n_{i_2}} \boldsymbol{\mu}_{i_1}^{\top} \mathbf{W} \boldsymbol{\Sigma}_{i_2} \mathbf{W} \boldsymbol{\mu}_{i_3} \right) + 2 \sum_{i_1 \neq i_2}  \beta_{i_1}^{2} \beta_{i_2}^{2}   \left( \frac{1}{n_{i_1} n_{i_2}}  \operatorname{tr}\left(  \mathbf{W} \boldsymbol{\Sigma}_{i_1}  \mathbf{W} \boldsymbol{\Sigma}_{i_2}\right) +  \boldsymbol{\mu}_{i_{1}}^{\top} \mathbf{W} \boldsymbol{\mu}_{i_{2}} \boldsymbol{\mu}_{i_{1}}^{\top} \mathbf{W} \boldsymbol{\mu}_{i_{2}} + \notag \right. \\
& \left. \frac{1}{ n_{i_2}} \boldsymbol{\mu}_{i_1}^{\top} \mathbf{W} \boldsymbol{\Sigma}_{i_2} \mathbf{W} \boldsymbol{\mu}_{i_1}  + \frac{1}{n_{i_1} }  \boldsymbol{\mu}_{i_2}^{\top} \mathbf{W} \boldsymbol{\Sigma}_{i_1} \mathbf{W} \boldsymbol{\mu}_{i_2}\right)  - \left( \sum_{i_{1} \neq i_{2}} \frac{\beta_{i_{1}} \beta_{i_{2}} }
{n_{i_{1}} n_{i_{2}}} \sum_{j_{1}, j_{2}} \boldsymbol{\mu}_{i_{1}}^{\top} \mathbf{W} \boldsymbol{\mu}_{i_{2}}  \right)^{2} \notag \\
=&   \sum_{i_1 \neq i_2 \neq i_3 \neq i_4} \beta_{i_1} \beta_{i_2} \beta_{i_3} \beta_{i_4} \boldsymbol{\mu}_{i_1}^{\top} \mathbf{W} \boldsymbol{\mu}_{i_2} \mathbf{W} \boldsymbol{\mu}_{i_3}^{\top} \boldsymbol{\mu}_{i_4}    + 4 \sum_{i_1 \neq i_2 \neq i_{3}} \beta_{i_1} \beta_{i_2} \beta_{i_3} \beta_{i_2}   \boldsymbol{\mu}_{i_2}^{\top} \mathbf{W} \boldsymbol{\mu}_{i_1} \boldsymbol{\mu}_{i_3}^{\top} \mathbf{W} \boldsymbol{\mu}_{i_2}                            \notag      \\
& + 2 \sum_{i_1 \neq i_2}  \beta_{i_1}^{2} \beta_{i_2}^{2}  \boldsymbol{\mu}_{i_{1}}^{\top} \mathbf{W} \boldsymbol{\mu}_{i_{2}} \boldsymbol{\mu}_{i_{1}}^{\top} \mathbf{W} \boldsymbol{\mu}_{i_{2}}  +  4 \sum_{i_1 \neq i_2 \neq i_{3}} \beta_{i_1} \beta_{i_2} \beta_{i_3} \beta_{i_2} \frac{1}{n_{i_2}}    \boldsymbol{\mu}_{i_1}^{\top} \mathbf{W} \boldsymbol{\Sigma}_{i_2} \mathbf{W} \boldsymbol{\mu}_{i_3}   \notag           \\ 
&+   2 \sum_{i_1 \neq i_2}  \beta_{i_1}^{2} \beta_{i_2}^{2}    \left( \frac{1}{ n_{i_2}}  \boldsymbol{\mu}_{i_1}^{\top} \mathbf{W} \boldsymbol{\Sigma}_{i_2} \mathbf{W} \boldsymbol{\mu}_{i_1}  + \frac{1}{ n_{i_1}} \boldsymbol{\mu}_{i_2}^{\top} \mathbf{W} \boldsymbol{\Sigma}_{i_1} \mathbf{W} \boldsymbol{\mu}_{i_2}\right) +2 \sum_{i_1 \neq i_2}  \beta_{i_1}^{2} \beta_{i_2}^{2} \frac{1}{n_{i_1} n_{i_2}}  \operatorname{tr}\left(  \mathbf{W} \boldsymbol{\Sigma}_{i_1}  \mathbf{W} \boldsymbol{\Sigma}_{i_2}\right)   \notag  \\
&- \left( \sum_{i_{1} \neq i_{2}} \frac{\beta_{i_{1}} \beta_{i_{2}} }
{n_{i_{1}} n_{i_{2}}} \sum_{j_{1}, j_{2}} \boldsymbol{\mu}_{i_{1}}^{\top} \mathbf{W} \boldsymbol{\mu}_{i_{2}}  \right)^{2} \notag\\
=&  4 \sum_{i_1 \neq i_2 \neq i_{3}} \beta_{i_1} \beta_{i_2} \beta_{i_3} \beta_{i_2} \frac{1}{n_{i_2}}    \boldsymbol{\mu}_{i_1}^{\top} \mathbf{W} \boldsymbol{\Sigma}_{i_2} \mathbf{W} \boldsymbol{\mu}_{i_3} +   2 \sum_{i_1 \neq i_2}  \beta_{i_1}^{2} \beta_{i_2}^{2}    \left( \frac{1}{ n_{i_2}}  \boldsymbol{\mu}_{i_1}^{\top} \mathbf{W} \boldsymbol{\Sigma}_{i_2} \mathbf{W} \boldsymbol{\mu}_{i_1}  + \frac{1}{ n_{i_1}} \boldsymbol{\mu}_{i_2}^{\top} \mathbf{W} \boldsymbol{\Sigma}_{i_1} \mathbf{W} \boldsymbol{\mu}_{i_2}\right)  \notag \\
&+2 \sum_{i_1 \neq i_2}  \beta_{i_1}^{2} \beta_{i_2}^{2} \frac{1}{n_{i_1} n_{i_2}}  \operatorname{tr}\left(  \mathbf{W} \boldsymbol{\Sigma}_{i_1}  \mathbf{W} \boldsymbol{\Sigma}_{i_2}\right)     \notag    \\
=& 4 \sum_{i_1 \neq i_2} \sum_{i_3 \neq i_2} \beta_{i_1} \beta_{i_2}^2 \beta_{i_3} \frac{1}{n_{i_2}} \boldsymbol{\mu}_{i_1}^{\top} \mathbf{W} \boldsymbol{\Sigma}_{i_2}\mathbf{W} \boldsymbol{\mu}_{i_3}  +2 \sum_{i_1 \neq i_2} \beta_{i_1}^2 \beta_{i_2}^2 \frac{1}{n_{i_1} n_{i_2}} \operatorname{tr}\left(  \mathbf{W} \boldsymbol{\Sigma}_{i_1}  \mathbf{W} \boldsymbol{\Sigma}_{i_2}\right) \label{varTn1}. 
\end{align}
Similarly, we can obtain
\begin{align}
&\operatorname{Var} \left( T_{n2} \right) \notag\\
= & \mathbb{E}\left(\sum_{i_1=1}^q \sum_{i_2=1}^q \frac{\beta_{i_1} ^2\beta_{i_2}^2}{n_{i_1}\left(n_{i_1}-1\right) n_{i_2}\left(n_{i_2}-1\right)} \sum_{j_1 \neq j_2} \sum_{j_3 \neq j_4}\left(\mathbf{x}_{i_1 j_1} \right)^{\top} \mathbf{W} \mathbf{x}_{i_1 j_2} \left(\mathbf{x}_{i_2 j_3} \right)^{\top} \mathbf{W} \mathbf{x}_{i_2 j_4}\right)     \notag\\
&-  \left(\sum_{i=1}^{q} \frac{\beta_{i}^{2}}{n_{i}\left(n_{i}-1\right)} \sum_{j_{1} \neq j_{2}}  \boldsymbol{\mu}_{i}^{\top} \mathbf{W} \boldsymbol{\mu}_{i}\right)^{2}    \notag\\
=& \sum_{i_1 \neq i_2}^q  \frac{\beta_{i_1} ^2\beta_{i_2}^2}{n_{i_1}\left(n_{i_1}-1\right) n_{i_2}\left(n_{i_2}-1\right)} \sum_{j_1 \neq j_2} \sum_{j_3 \neq j_4} \mathbb{E} \left(\mathbf{x}_{i_1 j_1} \right)^{\top} \mathbf{W} \mathbf{x}_{i_1 j_2} \left(\mathbf{x}_{i_2 j_3} \right)^{\top} \mathbf{W} \mathbf{x}_{i_2 j_4}    \notag\\
&+ \sum_{i_1=1}^q  \frac{\beta_{i_1}^4 }{n_{i_1}^2 \left(n_{i_1}-1 \right)^2  }  \sum_{j_1 \neq j_2} \sum_{j_3 \neq j_4} \mathbb{E} \left(\mathbf{x}_{i_1 j_1} \right)^{\top} \mathbf{W} \mathbf{x}_{i_1 j_2} \left(\mathbf{x}_{i_1 j_3} \right)^{\top} \mathbf{W} \mathbf{x}_{i_1 j_4}     \notag\\
&-  \left(\sum_{i=1}^{q} \frac{\beta_{i}^{2}}{n_{i}\left(n_{i}-1\right)} \sum_{j_{1} \neq j_{2}}  \boldsymbol{\mu}_{i}^{\top} \mathbf{W} \boldsymbol{\mu}_{i}\right)^{2}           \notag \\
=&  \sum_{i_1 \neq i_2}^q \beta_{i_1}^2\beta_{i_2} ^2 \boldsymbol{\mu}_{i_1}^{\top} \mathbf{W} \boldsymbol{\mu}_{i_1} \boldsymbol{\mu}_{i_2}^{\top} \mathbf{W} \boldsymbol{\mu}_{i_2}   +2 \sum_{i_1=1}^q  \frac{\beta_{i_1}^4 }{n_{i_1} \left(n_{i_1}-1 \right)  } \left( 2  \boldsymbol{\mu}_{i_1}^{\top}  \mathbf{W} \boldsymbol{\Sigma}_{i_1}  \mathbf{W} \boldsymbol{\mu}_{i_1} + \right.    \notag\\
& \left.  \left( \boldsymbol{\mu}_{i_1}^{\top} \mathbf{W} \boldsymbol{\mu}_{i_1}\right)^2 + \operatorname{tr}\left( \mathbf{W} \boldsymbol{\Sigma}_{i_1}\right)^2
\right) +  4 \sum_{i_1=1}^q  \frac{\beta_{i_1}^4 \left( n_{i_1} -2 \right) }{n_{i_1} \left(n_{i_1}-1 \right)  }  \left(    \boldsymbol{\mu}_{i_1}^{\top}  \mathbf{W} \boldsymbol{\Sigma}_{i_1}  \mathbf{W} \boldsymbol{\mu}_{i_1} +   \left( \boldsymbol{\mu}_{i_1}^{\top} \mathbf{W} \boldsymbol{\mu}_{i_1}\right)^2    \right)                       \notag   \\
&+ \sum_{i_1=1}^q  \frac{\beta_{i_1} ^4 \left( n_{i_1} -2 \right)\left( n_{i_1} -3 \right) }{n_{i_1} \left(n_{i_1}-1 \right)  } \left( \boldsymbol{\mu}_{i_1}^{\top} \mathbf{W} \boldsymbol{\mu}_{i_1}\right)^2      -  \left(\sum_{i=1}^{q} \frac{\beta_{i}^{2}}{n_{i}\left(n_{i}-1\right)} \sum_{j_{1} \neq j_{2}}  \boldsymbol{\mu}_{i}^{\top} \mathbf{W} \boldsymbol{\mu}_{i}\right)^{2}    \notag\\
= &\sum_{i_1=1}^q \sum_{i_2=1}^q\beta_{i_1}^2\beta_{i_2}^2 \boldsymbol{\mu}_{i_1}^{\top} \mathbf{W} \boldsymbol{\mu}_{i_1} \boldsymbol{\mu}_{i_2}^{\top} \mathbf{W} \boldsymbol{\mu}_{i_2}+4 \sum_{i_1=1}^q\beta_{i_1}^4 \frac{1}{n_{i_1}} \boldsymbol{\mu}_{i_1}^{\top} \mathbf{W} \boldsymbol{\Sigma}_{i_1} \mathbf{W}  \boldsymbol{\mu}_{i_1}    \notag\\
& +2 \sum_{i_1=1}^q\beta_{i_1}^4 \frac{1}{n_{i_1}\left(n_{i_1}-1\right)} \operatorname{tr}\left( \mathbf{W} \boldsymbol{\Sigma}_{i_1}\right)^2 -  \left(\sum_{i=1}^{q} \frac{\beta_{i}^{2}}{n_{i}\left(n_{i}-1\right)} \sum_{j_{1} \neq j_{2}}  \boldsymbol{\mu}_{i}^{\top} \mathbf{W} \boldsymbol{\mu}_{i}\right)^{2}    \notag \\
=& 4 \sum_{i=1}^q\beta_{i}^4 \frac{1}{n_{i}} \boldsymbol{\mu}_{i}^{\top} \mathbf{W} \boldsymbol{\Sigma}_{i} \mathbf{W}  \boldsymbol{\mu}_{i}  +2 \sum_{i=1}^q\beta_{i}^4 \frac{1}{n_{i}\left(n_{i}-1\right)} \operatorname{tr}\left( \mathbf{W} \boldsymbol{\Sigma}_{i}\right)^2    \label{varTn2}      , 
\end{align}
\begin{align}
&\operatorname{Cov} \left( T_{n1}, T_{n2} \right)\notag \\
=	&\mathbb{E}\left(\sum_{i_1 \neq i_2} \sum_{i_3=1}^q \frac{\beta_{i_1} \beta_{i_2}\beta_{i_3}^2}{n_{i_1} n_{i_2} n_{i_3}\left(n_{i_3}-1\right)} \sum_{j_1, j_2} \sum_{j_3 \neq j_4}\left(\mathbf{x}_{ i_1 j_1} \right)^{\top} \mathbf{W} \mathbf{x}_{i_2 j_2}\left(\mathbf{x}_{i_3 j_3} \right)^{\top} \mathbf{W} \mathbf{x}_{i_3 j_4}\right)  \notag \\
&- \left( \sum_{i_{1} \neq i_{2}} \frac{\beta_{i_{1}} \beta_{i_{2}} }
{n_{i_{1}} n_{i_{2}}} \sum_{j_{1}, j_{2}} \boldsymbol{\mu}_{i_{1}}^{\top} \mathbf{W} \boldsymbol{\mu}_{i_{2}}   \right) \left( \sum_{i=1}^{q} \frac{\beta_{i}^{2}}{n_{i}\left(n_{i}-1\right)} \sum_{j_{1} \neq j_{2}}  \boldsymbol{\mu}_{i}^{\top} \mathbf{W} \boldsymbol{\mu}_{i} \right)     \notag \\
=&  \sum_{i_{1} \neq i_{2} = i_{3} } \frac{ \beta_{i_{2}}^3 \beta_{i_{1}} } {n_{i_{2}}^{2} \left(n_{i_{2}}-1\right)  n_{i_{1}}}  \sum_{ j_{1}} \left( \sum_{ j_{2} =j_{3} \neq j_{4}} + \sum_{  j_{3} \neq j_{4} = j_{2}} + \sum_{  j_{2} \neq j_{3} \neq j_{4}} \right) \mathbb{E} \left(\mathbf{x}_{ i_{1} j_{1}} \right)^{\top} \mathbf{W} \mathbf{x}_{ i_{2} j_{2}}\left(\mathbf{x}_{ i_{2} j_{3}}\right)^{\top} \mathbf{W} \mathbf{x}_{i_{2} j_{4}}                       \notag     \\ 
&+  \sum_{ i_{3} =i_{1} \neq i_{2}} \frac{ \beta_{i_{1}}^3  \beta_{i_{2}} } {n_{i_{1}}^{2} \left(n_{i_{1}}-1\right)  n_{i_{2}}}  \sum_{ j_{2}} \left( \sum_{ j_{1} =j_{3} \neq j_{4}} + \sum_{  j_{3} \neq j_{4} = j_{1}} + \sum_{  j_{1} \neq j_{3} \neq j_{4}} \right) \mathbb{E} \left(\mathbf{x}_{ i_{1} j_{1}} \right)^{\top} \mathbf{W} \mathbf{x}_{ i_{2} j_{2}}\left(\mathbf{x}_{ i_{1} j_{3}}\right)^{\top} \mathbf{W} \mathbf{x}_{i_{1} j_{4}}       \notag   \\  
&- \left( \sum_{i_{1} \neq i_{2}} \frac{\beta_{i_{1}} \beta_{i_{2}} }
{n_{i_{1}} n_{i_{2}}} \sum_{j_{1}, j_{2}} \boldsymbol{\mu}_{i_{1}}^{\top} \mathbf{W} \boldsymbol{\mu}_{i_{2}}   \right) \left( \sum_{i=1}^{q} \frac{\beta_{i}^{2}}{n_{i}\left(n_{i}-1\right)} \sum_{j_{1} \neq j_{2}}  \boldsymbol{\mu}_{i}^{\top} \mathbf{W} \boldsymbol{\mu}_{i} \right)            \notag \\
=&  \sum_{i_{1} \neq i_{2}}  \beta_{i_{2}}^3  \beta_{i_{1}} \left(  \boldsymbol{\mu}_{i_{2}}^{\top}\mathbf{W} \boldsymbol{\mu}_{i_{2}} \boldsymbol{\mu}_{i_{2}}^{\top} \mathbf{W} \boldsymbol{\mu}_{i_{1}} + \frac{2}{n_{i_{2} } } \boldsymbol{\mu}_{i_{1}}^{\top} \mathbf{W} \boldsymbol{\Sigma}_{i_{2}} \mathbf{W}  \boldsymbol{\mu}_{i_{2}}          \right) + \sum_{i_{1} \neq i_{2}} \beta_{i_{1}}^3  \beta_{i_{2}}              \notag \\
&\cdot \left(  \boldsymbol{\mu}_{i_{1}}^{\top}\mathbf{W} \boldsymbol{\mu}_{i_{1}} \boldsymbol{\mu}_{i_{1}}^{\top} \mathbf{W} \boldsymbol{\mu}_{i_{2}} + \frac{ 2} {n_{i_{1}} } \boldsymbol{\mu}_{i_{1}}^{\top} \mathbf{W} \boldsymbol{\Sigma}_{i_{1}} \mathbf{W}  \boldsymbol{\mu}_{i_{2}}          \right)   - \left( \sum_{i_{1} \neq i_{2}} \frac{\beta_{i_{1}} \beta_{i_{2}} }
{n_{i_{1}} n_{i_{2}}} \sum_{j_{1}, j_{2}} \boldsymbol{\mu}_{i_{1}}^{\top} \mathbf{W} \boldsymbol{\mu}_{i_{2}}   \right)                      \notag\\
& \cdot \left( \sum_{i=1}^{q} \frac{\beta_{i}^{2}}{n_{i}\left(n_{i}-1\right)} \sum_{j_{1} \neq j_{2}}  \boldsymbol{\mu}_{i}^{\top} \mathbf{W} \boldsymbol{\mu}_{i} \right)                        \notag        \\
=&  \sum_{i_{1} \neq i_{2}}  \beta_{i_{2}}^3  \beta_{i_{1}} \left(   \frac{2}{n_{i_{2} } } \boldsymbol{\mu}_{i_{1}}^{\top} \mathbf{W} \boldsymbol{\Sigma}_{i_{2}} \mathbf{W}  \boldsymbol{\mu}_{i_{2}}          \right) + \sum_{i_{1} \neq i_{2}} \beta_{i_{1}}^3\beta_{i_{2}}  \left(  \frac{ 2} {n_{i_{1}} } \boldsymbol{\mu}_{i_{1}}^{\top} \mathbf{W} \boldsymbol{\Sigma}_{i_{1}} \mathbf{W}  \boldsymbol{\mu}_{i_{2}}          \right)  \notag \\
= & 4 \sum_{i_{1}\neq i_{2}}^{q}\beta_{i_{1}}^3\beta_{i_{2}} \frac{1}{n_{i_{1}}} \boldsymbol{\mu}_{i_{1}}^{\top} \mathbf{W} \boldsymbol{\Sigma}_{i_{1}}\mathbf{W} \boldsymbol{\mu}_{i_{2}}  \label{covTn1Tn2},
\end{align}
Finally, from \eqref{varTn1}-\eqref{covTn1Tn2}, we have
\begin{align*}
\operatorname{Var}\left(T_{n}\right)= &    4 \sum_{i_1 \neq i_2} \sum_{i_3 \neq i_2} \beta_{i_1} \beta_{i_2}^2 \beta_{i_3} \frac{1}{n_{i_2}} \boldsymbol{\mu}_{i_1}^{\top} \mathbf{W} \boldsymbol{\Sigma}_{i_2}\mathbf{W} \boldsymbol{\mu}_{i_3}  +2 \sum_{i_1 \neq i_2} \beta_{i_1}^2 \beta_{i_2}^2  \frac{1}{n_{i_1} n_{i_2}} \operatorname{tr}\left(  \mathbf{W} \boldsymbol{\Sigma}_{i_1}  \mathbf{W} \boldsymbol{\Sigma}_{i_2}\right)                      \\
&+ 4 \sum_{i_1=1}^q\beta_{i_1} ^2\beta_{i_1}^2 \frac{1}{n_{i_1}} \boldsymbol{\mu}_{i_1}^{\top} \mathbf{W} \boldsymbol{\Sigma}_{i_1} \mathbf{W}  \boldsymbol{\mu}_{i_1}  +2 \sum_{i=1}^q\beta_{i}^4 \frac{1}{n_{i}\left(n_{i}-1\right)} \operatorname{tr}\left( \mathbf{W} \boldsymbol{\Sigma}_{i}\right)^2  \\
&+  8 \sum_{i_{1}\neq i_{2}}^{q} \beta_{i_{1}}^3 \beta_{i_{2}} \frac{1}{n_{i_{1}}} \boldsymbol{\mu}_{i_{1}}^{\top} \mathbf{W} \boldsymbol{\Sigma}_{i_{1}}\mathbf{W} \boldsymbol{\mu}_{i_{2}}    \\	
\triangleq &  \sigma_{n,q_{1}}^{2} + \sigma_{n,q_{2}}^{2},
\end{align*}
where
\begin{align*}
\sigma_{n,q_{1}}^{2}=2 \sum_{i_{1} \neq i_{2}} \frac{\beta_{i_{1}} ^{2}\beta_{i_{2}}^{2}}{n_{i_{1}} n_{i_{2}}} \operatorname{tr}\left(  \mathbf{W} \boldsymbol{\Sigma}_{i_1}  \mathbf{W} \boldsymbol{\Sigma}_{i_2}\right)+2 \sum_{i=1}^{q} \frac{\beta_{i} ^{4}}{n_{i}\left(n_{i}-1\right)} \operatorname{tr}\left( \mathbf{W} \boldsymbol{\Sigma}_{i}\right)^2
\end{align*}
and
\begin{align*}
\sigma_{n,q_{2}}^{2}=4\left(\sum_{i=1}^{q} \beta_{i} \boldsymbol{\mu}_{i}\right)^{\top}  \mathbf{W}  \left(\sum_{i=1}^{q} \frac{\beta_{i}^2 }{n_{i}} \boldsymbol{\Sigma}_{i}\right)  \mathbf{W}  \left(\sum_{i=1}^{q} \beta_{i} \boldsymbol{\mu}_{i}\right).
\end{align*}
Note that
\begin{align*}
T_{n}   =&\sum_{i_{1} \neq i_{2}} \frac{\beta_{i_{1}} \beta_{i_{2}} }
{n_{i_{1}} n_{i_{2}}} \sum_{j_{1}, j_{2}}\left(\mathbf{x}_{i_{1} j_{1}}\right)^{\top} \mathbf{W} \left(\mathbf{x}_{i_{2} j_{2}} \right)
+\sum_{i=1} \frac{\beta_{i}^{2}}{n_{i}\left(n_{i}-1\right)} \sum_{j_{1} \neq j_{2}}\left(\mathbf{x}_{i j_{1}}\right)^{\top} \mathbf{W} \left(\mathbf{x}_{i j_{2}} \right)\\
=&\sum_{i_{1} \neq i_{2}} \frac{\beta_{i_{1}} \beta_{i_{2}} }
{n_{i_{1}} n_{i_{2}}} \sum_{j_{1}, j_{2}}\left( \mathbf{x}_{i_{1} j_{1}} - \boldsymbol{\mu}_{i_{1}} + \boldsymbol{\mu}_{i_{1}}\right)^{\top} \mathbf{W} \left(\mathbf{x}_{i_{2} j_{2}} - \boldsymbol{\mu}_{i_{2}} + \boldsymbol{\mu}_{i_{2}} \right) \\
&+\sum_{i=1} \frac{\beta_{i}^{2}}{n_{i}\left(n_{i}-1\right)} \sum_{j_{1} \neq j_{2}}\left( \mathbf{x}_{i j_{1}} - \boldsymbol{\mu}_{i} + \boldsymbol{\mu}_{i}\right)^{\top} \mathbf{W} \left( \mathbf{x}_{i j_{2}} - \boldsymbol{\mu}_{i} + \boldsymbol{\mu}_{i}\right)\\
=&\sum_{i_{1} \neq i_{2}} \frac{\beta_{i_{1}} \beta_{i_{2}} }
{n_{i_{1}} n_{i_{2}}} \sum_{j_{1}, j_{2}}\left( \mathbf{x}_{i_{1} j_{1}} - \boldsymbol{\mu}_{i_{1}} \right)^{\top} \mathbf{W} \left(\mathbf{x}_{i_{2} j_{2}} - \boldsymbol{\mu}_{i_{2}}  \right) + \sum_{i_{1} \neq i_{2}} \beta_{i_{1}} \beta_{i_{2}} \boldsymbol{\mu}_{i_{1}}  \mathbf{W} \boldsymbol{\mu}_{i_{2}} \\
&+\sum_{i=1} \frac{\beta_{i}^{2}}{n_{i}\left(n_{i}-1\right)} \sum_{j_{1} \neq j_{2}}\left( \mathbf{x}_{i j_{1}} - \boldsymbol{\mu}_{i} \right)^{\top} \mathbf{W} \left( \mathbf{x}_{i j_{2}} - \boldsymbol{\mu}_{i} \right)  +\sum_{i=1} \frac{\beta_{i}^{2}}{n_{i}\left(n_{i}-1\right)}   \boldsymbol{\mu}_{i}^{\top} \mathbf{W} \boldsymbol{\mu}_{i}  \\
&+ \sum_{i_{1} \neq i_{2}} \frac{\beta_{i_{1}} \beta_{i_{2}} }
{n_{i_{1}} n_{i_{2}}} \sum_{j_{1}, j_{2}} \left[ \left( \mathbf{x}_{i_{1} j_{1}} - \boldsymbol{\mu}_{i_{1}} \right)^{\top} \mathbf{W} \boldsymbol{\mu}_{i_{2}}  + \boldsymbol{\mu}_{i_{1}}^{\top} \mathbf{W} \left(\mathbf{x}_{i_{2} j_{2}} - \boldsymbol{\mu}_{i_{2}}  \right) \right]    +\sum_{i=1} \frac{\beta_{i}^{2}}{n_{i}\left(n_{i}-1\right)} \\
&  \sum_{j_{1} \neq j_{2}} \left[ \left( \mathbf{x}_{i j_{1}} - \boldsymbol{\mu}_{i} \right)^{\top} \mathbf{W} \boldsymbol{\mu}_{i} +  \boldsymbol{\mu}_{i}^{\top}\left( \mathbf{x}_{i j_{2}} - \boldsymbol{\mu}_{i} \right)        \right] \\
=&  \tilde{T}_{n} + \left( \sum_{i=1}^{q} \beta_{i} \boldsymbol{\mu}_{i} \right)^{\top} \mathbf{W} \left( \sum_{i=1}^{q} \beta_{i} \boldsymbol{\mu}_{i} \right),                         
\end{align*}
where 
\begin{align*}
\tilde{T}_{n}  =& \sum_{i_{1} \neq i_{2}} \frac{\beta_{i_{1}} \beta_{i_{2}} }
{n_{i_{1}} n_{i_{2}}} \sum_{j_{1}, j_{2}}\left( \mathbf{x}_{i_{1} j_{1}} - \boldsymbol{\mu}_{i_{1}} \right)^{\top} \mathbf{W} \left(\mathbf{x}_{i_{2} j_{2}} - \boldsymbol{\mu}_{i_{2}}  \right) + \sum_{i=1} \frac{\beta_{i}^{2}}{n_{i}\left(n_{i}-1\right)} \sum_{j_{1} \neq j_{2}}\left( \mathbf{x}_{i j_{1}} - \boldsymbol{\mu}_{i} \right)^{\top} \mathbf{W} \left( \mathbf{x}_{i j_{2}} - \boldsymbol{\mu}_{i} \right) \\ 
& +  \sum_{i_{1} \neq i_{2}} \frac{\beta_{i_{1}} \beta_{i_{2}} }
{n_{i_{1}} n_{i_{2}}} \sum_{j_{1}, j_{2}} \left[ \left( \mathbf{x}_{i_{1} j_{1}} - \boldsymbol{\mu}_{i_{1}} \right)^{\top} \mathbf{W} \boldsymbol{\mu}_{i_{2}}  + \boldsymbol{\mu}_{i_{1}}^{\top} \mathbf{W} \left(\mathbf{x}_{i_{2} j_{2}} - \boldsymbol{\mu}_{i_{2}}  \right) \right]    \\
&  +\sum_{i=1}^{q} \frac{\beta_{i}^{2}}{n_{i}\left(n_{i}-1\right)} \sum_{j_{1} \neq j_{2}} \left[ \left( \mathbf{x}_{i j_{1}} - \boldsymbol{\mu}_{i} \right)^{\top} \mathbf{W} \boldsymbol{\mu}_{i} +  \boldsymbol{\mu}_{i}^{\top}\mathbf{W}\left( \mathbf{x}_{i j_{2}} - \boldsymbol{\mu}_{i} \right)      \right]                                  \\
=& \sum_{i_{1} \neq i_{2}} \frac{\beta_{i_{1}} \beta_{i_{2}} }
{n_{i_{1}} n_{i_{2}}} \sum_{j_{1}, j_{2}}\left(  \boldsymbol{\Gamma}_{i_{1}} \mathbf{z}_{i_{1} j_{1}} \right)^{\top} \mathbf{W} \left(\boldsymbol{\Gamma}_{i_{2}} \mathbf{z}_{i_{2} j_{2}}  \right) +   \sum_{i=1} \frac{\beta_{i}^{2}}{n_{i}( n_{i} -1)} \sum_{j_{1} \neq j_{2}}\left( \boldsymbol{\Gamma}_{i} \mathbf{z}_{i j_{1}} \right)^{\top} \mathbf{W} \left( \boldsymbol{\Gamma}_{i} \mathbf{z}_{i j_{2}} \right)     \\
&+  \left( \sum_{ i_{2} =1}  \beta_{i_{2}} \boldsymbol{\mu}_{i_{2}} \right)^{\top} \mathbf{W}  \left(  \sum_{i_{1} =1}  \frac{ 2 \beta_{i_{1}}  }
{n_{i_{1}}}  \sum_{j_{1}}^{n_{i_{1}}} \boldsymbol{\Gamma}_{i_{1}} \mathbf{z}_{i_{1} j_{1}} \right),                       
\end{align*}
Therefore, it is sufficient to show that 
\begin{align*}
\frac{\tilde{T}_{n}  }{\sigma_{n,q_{1}}}	\stackrel{\text { d }}{\longrightarrow} \mathrm{N} (0,1).
\end{align*}
For convenience, define
$\mathbf{y}_{t+ \sum_{i=1}^{l-1} n_{i}}=\mathbf{W}^{1/2} \boldsymbol{\Gamma}_{i} \mathbf{z}_{ i j}$, $j=1, \ldots, n_{i}$, with $\sum_{i=1}^{0} n_{i} =0$, and
\begin{align*}
\begin{split}
&\varphi_{st}= \left \{
\begin{array}{ll}
\frac{\beta_{i}^{2}}{n_{i}(n_{i}-1)}\mathbf{y}_{s}^{\top}\mathbf{y}_{t}                   & s,t\in \Lambda_{i}, i=1,2,\cdots,q\\
\frac{\beta_{i} \beta_{l} }{n_{i} n_{l}}\mathbf{y}_{s}^{\top}\mathbf{y}_{t},   & (s,t)\in \Lambda_{i}\times \Lambda_{l}, 1\leq i<l \leq q
\end{array}
\right. \\
&\psi_{t} = \frac{1}{n_{i}} \beta_{i} \left(\sum_{i=1}^{q} \beta_{i} \boldsymbol{\mu}_{i}\right)^{\top}  \mathbf{W}^{1/2}\mathbf{y}_{t},
\end{split}
\end{align*}
where $t \in \Lambda_{i}=\{\sum\limits_{j=1}^{i-1}n_{j}+1,\sum\limits_{j=1}^{i-1}n_{j}+2, \cdots, \sum\limits_{j=1}^{i}n_{j}\}$
with $i=1, 2, \cdots, q$. In addition, we denote $c_{st}$ as the coefficient of $\varphi_{st}$, and if $ t\in \Lambda_{i}$, $\tilde{\boldsymbol{\Sigma}}_t=\boldsymbol{\Sigma}_i$ .

In what follows, we begin to construct square integrable martingale sequence. 
Denote  $V_{t}=\sum_{s=1}^{t -1} \varphi_{s t} +  \psi_{t}$ , and  $G_{m}=\sum_{t=1}^{m} V_{t}$,  and  $\mathcal{F}_{t}=\sigma\left\{\mathbf{y}_{1}, \mathbf{y}_{2} \ldots, \mathbf{y}_{t}\right\}$  which is the $\sigma$-field generated by  $\left\{\mathbf{y}_{1}, \mathbf{y}_{2} \ldots, \mathbf{y}_{t}\right\}$. Then, we have
\begin{align*}
\tilde{T}_{n}=&\sum_{i=1} \frac{\beta_{i}^{2}}{n_{i}(n_i-1)} \sum_{j_{1} \neq j_{2}}\left( \boldsymbol{\Gamma}_{i} \mathbf{z}_{i j_{1}} \right)^{\top} \mathbf{W} \left( \boldsymbol{\Gamma}_{i} \mathbf{z}_{i j_{2}} \right)
+\sum_{i_{1} \neq i_{2}} \frac{\beta_{i_{1}} \beta_{i_{2}} }{n_{i_{1}} n_{i_{2}}} \sum_{j_{1}, j_{2}}
\left(\boldsymbol{\Gamma}_{i_{1}} \mathbf{z}_{i_{1} j_{1}} \right)^{\top}\mathbf{W}\left(\boldsymbol{\Gamma}_{i_{2}} \mathbf{z}_{i_{2} j_{2}}\right)\\
&+2\sum_{i_{1} i_{2}}\frac{\beta_{i_{1}}\beta_{i_{2}}  }{n_{i_{1}}} \sum_{j_{1}}^{n_{i_{1}}}\boldsymbol{\mu}_{i_{2}}^{\top}\mathbf{W} (\boldsymbol{\Gamma}_{i_{1}} \mathbf{z}_{i_{1} j_{1}})\\
=&\sum_{i=1} \frac{\beta_{i}^{2}}{n_{i}(n_i-1)} 2 \sum_{j_{1}<j_{2}}\left( \boldsymbol{\Gamma}_{i} \mathbf{z}_{i j_{1}} \right)^{\top} \mathbf{W} \left( \boldsymbol{\Gamma}_{i} \mathbf{z}_{i j_{2}} \right)
+2\sum_{i_{1}<i_{2}} \frac{\beta_{i_{1}} \beta_{i_{2}} }{n_{i_{1}} n_{i_{2}}} \sum_{j_{1}, j_{2}}
\left(\boldsymbol{\Gamma}_{i_{1}} \mathbf{z}_{i_{1} j_{1}} \right)^{\top}\mathbf{W}\left(\boldsymbol{\Gamma}_{i_{2}} \mathbf{z}_{i_{2} j_{2}}\right)\\
&+2\sum_{i_{1}}\beta_{i_{1}}\sum_{j_{1}}\left(\sum_{i_{2}} \beta_{i_{2}} \boldsymbol{\mu}_{i_{2}}\right)^{\top}\mathbf{W}^{1/2}\left(\frac{1}{n_{i_{1}}}\mathbf{W}^{1/2}\boldsymbol{\Gamma}_{i_{1}} \mathbf{z}_{i_{1} j_{1}}\right)\\
=&2\sum_{i=1}^q\sum_{t=1+\sum_{m=1}^{i-1}n_m}^{\sum_{m=1}^{i}n_m}\sum_{s=1+\sum_{m=1}^{i-1}n_m}^{t-1}\varphi_{st}
+2\sum_{i=1}^{q-1}
\sum_{l=i+1}^q
\sum_{s=1+\sum_{m=1}^{i-1}n_m}^{\sum_{m=1}^{i}n_m}
\sum_{t=1+\sum_{m=1}^{l-1}n_m}^{\sum_{m=1}^{l}n_m}\phi_{st}
+2\sum_{i}\sum_{t\in \Lambda_{i}}\psi_{t}\\
=&2 \sum_{t=1}^{n} \left(\sum_{s=1}^{t-1} \varphi_{s t}\right)+2 \sum_{t=1}^{n} \psi_{t}\\
=&2 \sum_{t=1}^{n} V_{t}.
\end{align*}
In order to prove our main results, the following two lemmas are given.
\begin{lemma}
	For each $n$,  $\{G_{m}, \mathcal{F}_{m}\}_{m=1}^{n}$is the sequence of zero mean and a square integrable martingale.
	\begin{proof}
		By the definition of the martingale, it is obvious that  $\mathcal{F}_{m-1}\subseteq\mathcal{F}_{m}$, for any $1\leq m\leq n$ and $G_{m}$ is measurable with respect to  $\mathcal{F}_{m}$ .  In addition, note that
		\begin{align*}
		\mathbb{E}\left(V_{m} \mid \mathcal{F}_{m-1}\right) = 0,
		\end{align*}
		and
		\begin{align*}
		\mathbb{E} G_{m}^{2} \leq \operatorname{Var}\left(\tilde{T}_{n} \right)
		\end{align*}

		Next, we only need to show that $\{G_{m}\}$ is a martingale, that is, 
		$\mathbb{E}(G_{r}|\mathcal{F}_{m})=G_{m}$ for any $r\geq m$. Note that $\mathbb{E} (V_{\ell}|\mathcal{F}_{m}) =0$, $\forall \ell >m$. Then, 
		\begin{align*}
		\mathbb{E}(G_{r}|\mathcal{F}_{m})
		=G_{m}+\sum_{\ell=m+1}^{r}\mathbb{E}(V_{\ell}|\mathcal{F}_{m}) = G_{m}.
		\end{align*}
	\end{proof}	
\end{lemma}

\begin{lemma}
	Under \textbf{Assumption A-D}, as $n,p\to\infty,$  it gets
	\begin{align} \label{mean linear: the convergence of variance of CLT}
	\frac{\sum_{t=1}^{n} \mathbb{E}\left(V_{t}^{2} \mid \mathcal{F}_{t-1}\right)}{\operatorname{Var}\left(T_{n}\right)} \stackrel{\text { p }}{\longrightarrow} \frac{1}{4}
	\end{align}
	and
	\begin{align} \label{mean linear :the Lindeberg's condition of CLT}
	\frac{\sum_{t=1}^{n} \mathbb{E}\left(V_{t}^{2} I\left(\left|V_{t}\right|>\varepsilon \sqrt{\operatorname{Var}\left( T_{n} \right)}\right) \mid \mathcal{F}_{t-1}\right)}{\operatorname{Var}\left(T_{n}\right)} \stackrel{\text { p }}{\longrightarrow} 0,
	\end{align}	
where  $\stackrel{\text { p }} {\longrightarrow}$ means convergence in probability.
\end{lemma}
\begin{proof}

	Denote  
	\begin{align*}
	\zeta_{n}=\sum_{t=1}^{n} \mathbb{E}\left(V_{t}^{2} \mid \mathcal{F}_{t-1}\right)
	=\sum_{i=1}^{q}\sum_{t\in \Lambda_{i}} \mathbb{E}\left(V_{t}^{2} \mid \mathcal{F}_{t-1}\right).
	\end{align*} 		
	To begin with, for $ t\in \Lambda_{i}$, we have 
	\begin{align*}
	\mathbb{E}\left(V_{t}^{2} \mid \mathcal{F}_{t-1}\right)
	=&\mathbb{E}\left[\left(\sum_{s=1}^{t -1} \varphi_{s t}\right)^{2}+(\psi_{t})^{2}+2\sum_{s=1}^{t -1} \varphi_{st}\psi_{t}\mid \mathcal{F}_{t-1}\right] \\
	=&\mathbb{E}\left[\sum_{s_1,s_2=1}^{t-1}\varphi_{s_{1}t}\varphi_{s_{2}t}
	+\frac{\beta_{i}^{2}}{n_{i}^2}(\sum_{i=1}^{q} \beta_{i}\boldsymbol{\mu}_{i})^{\top}\mathbf{W}^{1/2}\mathbf{y}_{t}\mathbf{y}_{t}^{\top}\mathbf{W}^{1/2}(\sum_{i=1}^{q} \beta_{i}\boldsymbol{\mu}_{i})
	\right.\\
	&\left.+2\sum_{s=1}^{t -1} \varphi_{st}\frac{\beta_{i}}{n_{i}} \left(\sum_{i=1}^{q} \beta_{i} \boldsymbol{\mu}_{i}\right)^{\top}  \mathbf{W}^{1/2}\mathbf{y}_{t} \mid \mathcal{F}_{t-1}
	\right] \\
	=&\sum_{s_1,s_2=1}^{t-1}c_{s_1t}c_{s_2t}\mathbf{y}_{s_1}^{\top}\mathbb{E}(\mathbf{y}_t\mathbf{y}_t^{\top}|\mathcal{F}_{t-1})\mathbf{y}_{s_2}
	+\frac{\beta_{i}^{2}}{n_{i}^2}(\sum_{i=1}^{q} \beta_{i}\boldsymbol{\mu}_{i})^{\top}\mathbf{W}^{1/2}\mathbb{E}(\mathbf{y}_{t}\mathbf{y}_{t}^{\top}|\mathcal{F}_{t-1})\mathbf{W}^{1/2}(\sum_{i=1}^{q} \beta_{i}\boldsymbol{\mu}_{i})\\
	&+2\mathbb{E}\left(\sum_{s=1}^{t -1} c_{st}\frac{\beta_{i}}{n_{i}}\mathbf{y}_{s}^{\top}\mathbf{y}_{t} \left(\sum_{i=1}^{q} \beta_{i} \boldsymbol{\mu}_{i}\right)^{\top}  \mathbf{W}^{1/2}\mathbf{y}_{t}\mid \mathcal{F}_{t-1}\right)\\
	=&\sum_{s_1,s_2=1}^{t-1}c_{s_1t}c_{s_2t}\mathbf{y}_{s_1}^{\top}\mathbb{E}(\mathbf{y}_t\mathbf{y}_t^{\top})\mathbf{y}_{s_2}
	+\frac{\beta_{i}^{2}}{n_{i}^2}(\sum_{i=1}^{q} \beta_{i}\boldsymbol{\mu}_{i})^{\top}\mathbf{W}^{1/2}\mathbb{E}(\mathbf{y}_{t}\mathbf{y}_{t}^{\top})\mathbf{W}^{1/2}(\sum_{i=1}^{q} \beta_{i}\boldsymbol{\mu}_{i})\\
	&+2\sum_{s=1}^{t -1}c_{st}\frac{\beta_{i}}{n_{i}} \mathbf{y}_{s}^{\top} \mathbb{E}\left(\mathbf{y}_{t} (\sum_{i=1}^{q} \beta_{i} \boldsymbol{\mu}_{i} )^{\top} \mathbf{W}^{1/2}\mathbf{y}_{t}^{\top} \right)\\
	=&\sum_{s_1,s_2=1}^{t-1}c_{s_1t}c_{s_2t}\mathbf{y}_{s_1}^{\top}\mathbf{W}^{1/2}\tilde{\boldsymbol{\Sigma}}_{t}\mathbf{W}^{1/2}\mathbf{y}_{s_2}
	+\frac{\beta_{i}^{2}}{n_{i}^2}(\sum_{i=1}^{q} \beta_{i}\boldsymbol{\mu}_{i})^{\top}\mathbf{W}\tilde{\boldsymbol{\Sigma}}_{t}\mathbf{W}
	(\sum_{i=1}^{q} \beta_{i}\boldsymbol{\mu}_{i})\\
	&+2\sum_{s=1}^{t -1}c_{st}\frac{\beta_{i}}{n_{i}}\mathbf{y}_{s}^{\top}
	\mathbf{W}^{1/2} \tilde{\boldsymbol{\Sigma}}_{t}	\mathbf{W} \left(\sum_{i=1}^{q} \beta_{i} \boldsymbol{\mu}_{i}\right)
	\end{align*}

	Then, it follows that
	\begin{align*}
	&\mathbb{E}(\zeta_{n})\\
	=&\sum_{t=1}^{n}\mathbb{E}\left[\sum_{s_1,s_2=1}^{t-1}c_{s_1t}c_{s_2t}\mathbf{y}_{s_1}^{\top}\mathbf{W}^{1/2}\tilde{\boldsymbol{\Sigma}}_{t}\mathbf{W}^{1/2}\mathbf{y}_{s_2}
	+\frac{\beta_{i}^{2}}{n_{i}^2}(\sum_{i=1}^{q} \beta_{i}\boldsymbol{\mu}_{i})^{\top}\mathbf{W}\tilde{\boldsymbol{\Sigma}}_{t}\mathbf{W}
	(\sum_{i=1}^{q} \beta_{i}\boldsymbol{\mu}_{i})\right]\\
	=&\sum_{i=1}^{q}\sum_{t\in \Lambda_{i}}\sum_{s_1= s_2}^{t-1}\mathbb{E}\left(c_{s_1t}^{2}\mathbf{y}_{s_1}^{\top}\mathbf{W}^{1/2}\tilde{\boldsymbol{\Sigma}}_{t}\mathbf{W}^{1/2}\mathbf{y}_{s_1}\right)
	+\sum_{i=1}^{q}\sum_{t\in \Lambda_{i}} \frac{\beta_{i}^{2}}{n_{i}^2}(\sum_{i=1}^{q} \beta_{i}\boldsymbol{\mu}_{i})^{\top}\mathbf{W}\tilde{\boldsymbol{\Sigma}}_{t}\mathbf{W}
	(\sum_{i=1}^{q} \beta_{i}\boldsymbol{\mu}_{i})\\
	=& \sum_{i<l}\frac{\beta_{i}^{2} \beta_{l}^{2} }{n_{l} n_{i}} \operatorname{tr} (\mathbf{W}\boldsymbol{\Sigma}_{i } \mathbf{W}\boldsymbol{\Sigma}_{l})+\sum_{i=1} \frac{\beta_{i}^{4}}{2n_{i}(n_{i}-1)}\operatorname{tr}( \mathbf{W}\boldsymbol{\Sigma}_{i} )^{2} +(\sum_{i=1}^{q} \beta_{i}\boldsymbol{\mu}_{i})^{\top}\mathbf{W}^{1/2}(\sum_{i=1}^{q} \frac{\beta_{i}^{2}}{n_{i}}\boldsymbol{\Sigma}_{i} )\mathbf{W}^{1/2}
	(\sum_{i=1}^{q} \beta_{i}\boldsymbol{\mu}_{i})\\
	=&\frac{1}{4}\operatorname{Var}T_{n}.
	\end{align*}
	Next, we split the term $\mathbb{E}( \zeta_{n}^{2})$ into six term, and analyze them one by one.			
	\begin{align}
	\mathbb{E}( \zeta_{n}^{2}) =& \mathbb{E}\left[\sum_{t=1}^{n} \sum_{s_1,s_2=1}^{t-1}c_{s_1t}c_{s_2t}\mathbf{y}_{s_1}^{\top}\mathbf{W}^{1/2}\tilde{\boldsymbol{\Sigma}}_{t}\mathbf{W}^{1/2}\mathbf{y}_{s_2}
	+\sum_{t=1}^{n}\frac{\beta_{i}^{2}}{n_{i}^2}(\sum_{i=1}^{q} \beta_{i}\boldsymbol{\mu}_{i})^{\top}\mathbf{W}\tilde{\boldsymbol{\Sigma}}_{t}\mathbf{W}
	(\sum_{i=1}^{q} \beta_{i}\boldsymbol{\mu}_{i}) \right. \notag\\
	&\left.+2\sum_{t=1}^{n}\sum_{s=1}^{t -1}c_{st}\frac{\beta_{i}}{n_{i}}\mathbf{y}_{s}^{\top}
	\mathbf{W}^{1/2} \tilde{\boldsymbol{\Sigma}}_{t}	\mathbf{W} \left(\sum_{i=1}^{q} \beta_{i} \boldsymbol{\mu}_{i}\right)\right]^{2} \notag\\
	=&: K_{1} + K_{2} + 4K_{3} +2K_{4} +4K_{5} + 0. \label{the spiliting of two moments of zetan}
	\end{align}	
	After tedious calculations, we can get the following results:
	\begin{align}
	K_{1} =& \sum_{t_{1},t_{2}}^{n}\sum_{s_1, s_2}^{t_{1}-1}\sum_{s_3, s_4}^{t_{2}-1}c_{s_1t_{1}}c_{s_2t_{1}}c_{s_3t_{2}}c_{s_4t_{2}}\mathbb{E}\mathbf{y}_{s_1}^{\top}\mathbf{W}^{1/2}\tilde{\boldsymbol{\Sigma}}_{t_{1}}\mathbf{W}^{1/2}\mathbf{y}_{s_2}\mathbf{y}_{s_3}^{\top}\mathbf{W}^{1/2}\tilde{\boldsymbol{\Sigma}}_{t_{2}}\mathbf{W}^{1/2}\mathbf{y}_{s_4}\notag\\
	=& \sum_{t_{1}\neq t_{2}}^{n} \sum_{s_1= s_2}^{t_{1}-1}\sum_{s_3 = s_4}^{t_{2}-1} c_{s_1t_{1}}^{2}c_{s_3t_{2}}^{2} \mathbb{E}\mathbf{y}_{s_1}^{\top}\mathbf{W}^{1/2}\tilde{\boldsymbol{\Sigma}}_{t_{1}}\mathbf{W}^{1/2}\mathbf{y}_{s_1}\mathbf{y}_{s_3}^{\top}\mathbf{W}^{1/2}\tilde{\boldsymbol{\Sigma}}_{t_{2}}\mathbf{W}^{1/2}\mathbf{y}_{s_3}                                \notag      \\
	& +\sum_{t_{1}= t_{2}}^{n} \sum_{s_1, s_2, s_3, s_4}^{t_{1}-1}c_{s_1t_{1}}c_{s_2t_{1}}c_{s_3t_{1}}c_{s_4t_{1}} \mathbb{E}\mathbf{y}_{s_1}^{\top}\mathbf{W}^{1/2}\tilde{\boldsymbol{\Sigma}}_{t_{1}}\mathbf{W}^{1/2}\mathbf{y}_{s_2}\mathbf{y}_{s_3}^{\top}\mathbf{W}^{1/2}\tilde{\boldsymbol{\Sigma}}_{t_{1}}\mathbf{W}^{1/2}\mathbf{y}_{s_4}                                  \notag  \\
	=& \sum_{i\neq l}^{q} \left( \sum_{t_{1}\in \Lambda_{i}} \sum_{s_1}^{t_{1}-1}c_{s_1t_{1}}^{2} \mathbb{E}\mathbf{y}_{s_1}^{\top}\mathbf{W}^{1/2}\tilde{\boldsymbol{\Sigma}}_{t_{1}}\mathbf{W}^{1/2}\mathbf{y}_{s_1}  \right)  \left( \sum_{t_{2}\in \Lambda_{l}}\sum_{s_3}^{t_{2}-1}c_{s_3t_{2}}^{2} \mathbb{E} \mathbf{y}_{s_3}^{\top}\mathbf{W}^{1/2}\tilde{\boldsymbol{\Sigma}}_{t_{2}}\mathbf{W}^{1/2}\mathbf{y}_{s_3}  \right)        \notag       \\
	&+ \sum_{i=1}^{q} \sum_{t\in \Lambda_{i}} \sum_{s_2=s_1\neq s_3 =s_4}^{t-1} c_{s_1t_{1}}^{2}c_{s_3t_{1}}^{2} \left( \mathbb{E}\mathbf{y}_{s_1}^{\top}\mathbf{W}^{1/2}\tilde{\boldsymbol{\Sigma}}_{t}\mathbf{W}^{1/2}\mathbf{y}_{s_1} \right)  \left(\mathbb{E} \mathbf{y}_{s_3}^{\top}\mathbf{W}^{1/2}\tilde{\boldsymbol{\Sigma}}_{t}\mathbf{W}^{1/2}\mathbf{y}_{s_3} \right)                           \notag                  \\
	&+ 2 \sum_{i=1}^{q} \sum_{t\in \Lambda_{i}} \sum_{s_3=s_1\neq s_2 =s_4}^{t-1} c_{s_1t_{1}}^{2}c_{s_2t_{1}}^{2} \left( \mathbb{E}\mathbf{y}_{s_1}^{\top}\mathbf{W}^{1/2}\tilde{\boldsymbol{\Sigma}}_{t}\mathbf{W}^{1/2}\mathbf{y}_{s_2}  \mathbf{y}_{s_1}^{\top}\mathbf{W}^{1/2}\tilde{\boldsymbol{\Sigma}}_{t}\mathbf{W}^{1/2}\mathbf{y}_{s_2} \right)                                     \notag        \\
	&+ \sum_{i=1}^{q} \sum_{t\in \Lambda_{i}} \sum_{s=1}^{t-1} c_{s t}^{4} \mathbb{E}\left( \mathbf{y}_{s}^{\top}\mathbf{W}^{1/2}\tilde{\boldsymbol{\Sigma}}_{t}\mathbf{W}^{1/2}\mathbf{y}_{s}\mathbf{y}_{s}^{\top}\mathbf{W}^{1/2}\tilde{\boldsymbol{\Sigma}}_{t}\mathbf{W}^{1/2}\mathbf{y}_{s} \right)                        \notag           \\
	= & \frac{1}{16}\sigma_{n,q_{1}}^{4} (1 + o(1)) + O(n^{-4}) \sum_{i_1, i_2, i_3}^{q}  \operatorname{tr} (\mathbf{W}\boldsymbol{\Sigma}_{i_1 } \mathbf{W}\boldsymbol{\Sigma}_{i_3} \mathbf{W}\boldsymbol{\Sigma}_{i_2} \mathbf{W}\boldsymbol{\Sigma}_{i_3})  \notag  \\
	&+ O(n^{-5}) \left(  \sum_{i_1, i_2}^{q}\operatorname{tr} (\mathbf{W}\boldsymbol{\Sigma}_{i_1 } \mathbf{W}\boldsymbol{\Sigma}_{i_2} \mathbf{W}\boldsymbol{\Sigma}_{i_1 } \mathbf{W}\boldsymbol{\Sigma}_{i_2}) +\sum_{i_1, i_2, i_3, i_4}^{q} \operatorname{tr} (\mathbf{W}\boldsymbol{\Sigma}_{i_1 } \mathbf{W}\boldsymbol{\Sigma}_{i_2})\operatorname{tr} (\mathbf{W}\boldsymbol{\Sigma}_{i_3 } \mathbf{W}\boldsymbol{\Sigma}_{i_4}) \right)    
	\notag\\  
	= & \frac{1}{16}\sigma_{n,q_{1}}^{4} (1 + o(1)) + o\left(  \operatorname{Var}^{2}\left(T_{n}\right) \right), \label{K1:the spiliting of two moments of zetan}   
	\end{align}
	\begin{align}
	K_{2} =&  \left[ (\sum_{i=1}^{q} \beta_{i}\boldsymbol{\mu}_{i})^{\top}\mathbf{W}^{1/2}(\sum_{i=1}^{q} \frac{\beta_{i}^{2}}{n_{i}}\boldsymbol{\Sigma}_{i} )\mathbf{W}^{1/2}
	(\sum_{i=1}^{q} \beta_{i}\boldsymbol{\mu}_{i}) \right]^{2} = \frac{1}{16}  \sigma_{n,q_{2}}^{4},   \label{K2:the spiliting of two moments of zetan}   
	\end{align}
	\begin{align}
	K_{3} =& \sum_{t_{1}, t_{2}}^{n}\sum_{s_1}^{t_{1}-1}\sum_{s_2}^{t_{2}-1}c_{s_1t_{1}}c_{s_2t_{2}}\frac{\beta_{i}^{2}}{n_{i}^{2}} \mathbb{E} \left( \mathbf{y}_{s_{1}}^{\top}
	\mathbf{W}^{1/2} \tilde{\boldsymbol{\Sigma}}_{t_{1}}	\mathbf{W} (\sum_{i=1}^{q} \beta_{i} \boldsymbol{\mu}_{i})(\sum_{i=1}^{q} \beta_{i} \boldsymbol{\mu}_{i})^{\top}\mathbf{W} \tilde{\boldsymbol{\Sigma}}_{t_{2}}\mathbf{y}_{s_{2}}\right)           \notag\\
	=&O(n^{-6})\sum_{t=1}^{n} \sum_{s}^{t-1}\mathbb{E}  \left( \mathbf{y}_{s}^{\top}
	\mathbf{W}^{1/2} \tilde{\boldsymbol{\Sigma}}_{t}	\mathbf{W} (\sum_{i=1}^{q} \beta_{i} \boldsymbol{\mu}_{i}) (\sum_{i=1}^{q} \beta_{i} \boldsymbol{\mu}_{i})^{\top}\mathbf{W}  \tilde{\boldsymbol{\Sigma}}_{t} \mathbf{W}^{1/2} \mathbf{y}_{s}\right)                   \notag \\
	=&O(n^{-4}) \sum_{l \leq i}^{q}  (\sum_{i=1}^{q} \beta_{i} \boldsymbol{\mu}_{i})^{\top}\mathbf{W}  \boldsymbol{\Sigma}_{i} \mathbf{W}\boldsymbol{\Sigma}_{l} \mathbf{W} \boldsymbol{\Sigma}_{i}	\mathbf{W} (\sum_{i=1}^{q} \beta_{i} \boldsymbol{\mu}_{i})             \notag\\
	=&O(n^{-4}) \sum_{l \leq i} (\sum_{i=1}^{q} \beta_{i} \boldsymbol{\mu}_{i})^{\top}\mathbf{W}^{1/2}(\mathbf{W}^{1/2} \boldsymbol{\Sigma}_{i} \mathbf{W}^{1/2})
	(\mathbf{W}^{1/2} \boldsymbol{\Sigma}_{l} \mathbf{W}^{1/2})(\mathbf{W}^{1/2} \boldsymbol{\Sigma}_{i} \mathbf{W}^{1/2})\mathbf{W}^{1/2}(\sum_{i=1}^{q} \beta_{i} \boldsymbol{\mu}_{i})                           \notag\\
	\leq & O(n^{-4}) \sum_{l \leq i}(\sum_{i=1}^{q} \beta_{i} \boldsymbol{\mu}_{i})^{\top}\mathbf{W}\boldsymbol{\Sigma}_{i}\mathbf{W}(\sum_{i=1}^{q} \beta_{i} \boldsymbol{\mu}_{i}) \lambda_{max}(\mathbf{W} \boldsymbol{\Sigma}_{l} \mathbf{W}\boldsymbol{\Sigma}_{i})                     \notag\\
	\leq & O(n^{-4})\sum_{l \leq i}(\sum_{i=1}^{q} \beta_{i} \boldsymbol{\mu}_{i})^{\top}\mathbf{W}\boldsymbol{\Sigma}_{i}\mathbf{W}(\sum_{i=1}^{q} \beta_{i} \boldsymbol{\mu}_{i}) \operatorname{tr}(\mathbf{W} \boldsymbol{\Sigma}_{l} \mathbf{W}\boldsymbol{\Sigma}_{i})              \notag\\
	=&o\left(  \operatorname{Var}^{2}\left(T_{n}\right) \right).\label{K3:the spiliting of two moments of zetan}   
	\end{align}		
	\begin{align}
	K_{4} =& \left( \sum_{i=1}^{q}\sum_{t\in \Lambda_{i}}\sum_{s_1,s_2=1}^{t-1}\mathbb{E}\left(c_{s_1t}c_{s_2t}\mathbf{y}_{s_1}^{\top}\mathbf{W}^{1/2}\tilde{\boldsymbol{\Sigma}}_{t}\mathbf{W}^{1/2}\mathbf{y}_{s_2}\right) \right) \times \left(  (\sum_{i=1}^{k} \beta_{i}\boldsymbol{\mu}_{i})^{\top}\mathbf{W}^{1/2}(\sum_{i=1}^{q} \frac{\beta_{i}^{2}}{n_{i}}\boldsymbol{\Sigma}_{i} )\mathbf{W}^{1/2}
	(\sum_{i=1}^{k} \beta_{i}\boldsymbol{\mu}_{i})\right) \notag \\
	= & \frac{1}{16} \sigma_{n,q_{1}}^{2}  \sigma_{n,q_{2}}^{2}.\label{K4:the spiliting of two moments of zetan}   
	\end{align}

	\begin{align}
	K_{5} =& \sum_{t_{1},t_{2}}^{n}\sum_{s_1, s_2}^{t_{1}-1}\sum_{s_3}^{t_{2}-1}c_{s_1t_{1}}c_{s_2t_{1}}c_{s_3t_{2}}\frac{\beta_{i}}{n_{i}} \mathbb{E} \left( \mathbf{y}_{s_1}^{\top}\mathbf{W}^{1/2}\tilde{\boldsymbol{\Sigma}}_{t_{1}}\mathbf{W}^{1/2}\mathbf{y}_{s_2}\mathbf{y}_{s_{3}}^{\top}
	\mathbf{W}^{1/2} \tilde{\boldsymbol{\Sigma}}_{t_{2}}	\mathbf{W} (\sum_{i=1}^{q} \beta_{i} \boldsymbol{\mu}_{i}) \right)             \notag \\
	=&\sum_{t_{1}=t_{2}}^{n} \sum_{s_1, s_2, s_3}^{t_{1}-1}c_{s_1t_{1}}c_{s_2t_{1}}c_{s_3t_{1}}\frac{\beta_{i}}{n_{i}} \mathbb{E} \left( \mathbf{y}_{s_1}^{\top}\mathbf{W}^{1/2}\tilde{\boldsymbol{\Sigma}}_{t_{1}}\mathbf{W}^{1/2}\mathbf{y}_{s_2}\mathbf{y}_{s_{3}}^{\top}
	\mathbf{W}^{1/2} \tilde{\boldsymbol{\Sigma}}_{t_{1}}	\mathbf{W} (\sum_{i=1}^{q} \beta_{i} \boldsymbol{\mu}_{i}) \right)                \notag\\
	=&\sum_{t_{1}=t_{2}}^{n} \sum_{s_1}^{t_{1}-1}c_{s_1t_{1}}^{3}\frac{\beta_{i}}{n_{i}} \mathbb{E} \left( \mathbf{y}_{s_1}^{\top}\mathbf{W}^{1/2}\tilde{\boldsymbol{\Sigma}}_{t_{1}}\mathbf{W}^{1/2}\mathbf{y}_{s_1}\mathbf{y}_{s_{1}}^{\top}
	\mathbf{W}^{1/2} \tilde{\boldsymbol{\Sigma}}_{t_{1}}	\mathbf{W} (\sum_{i=1}^{q} \beta_{i} \boldsymbol{\mu}_{i}) \right)      \notag\\
	=& O(n^{-7}) \sum_{t_{1}}^{n} \sum_{s_1}^{t_{1}-1}\mathbb{E}^{\frac{1}{2}} \left( \mathbf{y}_{s_1}^{\top}\mathbf{W}^{1/2}\tilde{\boldsymbol{\Sigma}}_{t_{1}}\mathbf{W}^{1/2}\mathbf{y}_{s_1}\right)^{2}    \mathbb{E}^{\frac{1}{2}} \left( \mathbf{y}_{s_{1}}^{\top}
	\mathbf{W}^{1/2} \tilde{\boldsymbol{\Sigma}}_{t_{1}}	\mathbf{W} (\sum_{i=1}^{q} \beta_{i} \boldsymbol{\mu}_{i}) \right)^{2}                \notag       \\
	=& O(n^{-5}) \mathbf{1}^{\top}\operatorname{diag}\left( \boldsymbol{\Gamma}_{1}^{\top}\mathbf{W}\boldsymbol{\Sigma}_{2}\mathbf{W}\boldsymbol{\Gamma}_{1} \right)\boldsymbol{\Gamma}_{1}^{\top}\mathbf{W}\boldsymbol{\Sigma}_{2}\mathbf{W}(\sum_{i=1}^{q} \beta_{i} \boldsymbol{\mu}_{i}) \notag\\
	\leq&O(n^{-5}) \left( \mathbf{1}^{\top}\operatorname{diag}^{2}\left( \boldsymbol{\Gamma}_{1}^{\top}\mathbf{W}\boldsymbol{\Sigma}_{2}\mathbf{W}\boldsymbol{\Gamma}_{1} \right) \mathbf{1}\right)^{1/2}\left((\sum_{i=1}^{q} \beta_{i} \boldsymbol{\mu}_{i})^{\top} \mathbf{W}\boldsymbol{\Sigma}_{2}\mathbf{W}\boldsymbol{\Sigma}_{1}\mathbf{W}\boldsymbol{\Sigma}_{2}\mathbf{W} (\sum_{i=1}^{q} \beta_{i} \boldsymbol{\mu}_{i}) \right)^{1/2} \notag\\
	\leq &O(n^{-5})\operatorname{tr}\left(\mathbf{W}\boldsymbol{\Sigma}_{2}\mathbf{W}\boldsymbol{\Sigma}_{1}\right) \left( (\sum_{i=1}^{q} \beta_{i} \boldsymbol{\mu}_{i})^{\top}\mathbf{W}\boldsymbol{\Sigma}_{1}\mathbf{W}(\sum_{i=1}^{q} \beta_{i} \boldsymbol{\mu}_{i}) \operatorname{tr}(\mathbf{W} \boldsymbol{\Sigma}_{2} \mathbf{W}\boldsymbol{\Sigma}_{1})\right)^{1/2} \notag\\
	=&o\left(  \operatorname{Var}^{2}\left(T_{n}\right) \right)\label{K5:the spiliting of two moments of zetan}   
	\end{align}
	where we used the identity
	$
	\mathbb{E} \left( \mathbf{u}^{\top} \mathbf{A} \mathbf{u} \cdot \mathbf{u}^{\top} \mathbf{C} \mathbf{r}  \right) = \sum_{i=1}^{p}  \mathbb{E}u_1^{3} a_{ii}c_{ii} r_{i} + \sum_{i \ne j}  \mathbb{E}u_1^{3} a_{ii}u_{ij} r_{j} = \mathbb{E} u_1^{3} \cdot \mathbf{1}^{\top}\operatorname{diag}(\mathbf{A}) \mathbf{C} \mathbf{r}$, $\mathbf{u}=(u_1,\ldots,u_p)^{\top}$ is a random vector with i.i.d. entries satisfying $\mathbb{E} x_1=0, \mathbb{E} x_1^2=1$, $\mathbf{r}=(r_1,\ldots,r_p)^{\top}$ is a non-random vector, and $\mathbf{A}= (a_{ij}),\mathbf{C}= (c_{ij}) $ be complex $p \times p$ nonrandom matrices.

	By combining \eqref{the spiliting of two moments of zetan}, \eqref{K1:the spiliting of two moments of zetan}, \eqref{K2:the spiliting of two moments of zetan}, \eqref{K3:the spiliting of two moments of zetan}, \eqref{K4:the spiliting of two moments of zetan}, and \eqref{K5:the spiliting of two moments of zetan}, we obtain
	\begin{align*}
	\operatorname{Var}(\zeta_{n}) =\mathbb{E}( \zeta_{n}^{2}) -\mathbb{E}^{2}( \zeta_{n}) =o\left(  \operatorname{Var}^{2}\left(T_{n}\right) \right).
	\end{align*}
	which completes the proof of \eqref{mean linear: the convergence of variance of CLT}.

	Note that
	\begin{align}
	\frac{\sum_{t=1}^{n} \mathbb{E}\left(V_{t}^{2} I\left(\left|V_{t}\right|>\varepsilon \sqrt{\operatorname{Var}\left( T_{n} \right)}\right) \mid \mathcal{F}_{t-1}\right)}{\operatorname{Var}\left(T_{n}\right)} \leq 
	\varepsilon^{-2} \frac{\sum_{t=1}^{n} \mathbb{E}\left(V_{t}^{4} \mid \mathcal{F}_{t-1}\right)}{\operatorname{Var}^{2}\left(T_{n}\right)} .
	\end{align}	
	Hence, to prove \eqref{mean linear :the Lindeberg's condition of CLT}, it suffices to show that
	\begin{align} \label{equivalent condition of the Lindeberg's condition of CLT}
	\mathbb{E}\left(\sum_{t=1}^{n}\mathbb{E}(V_{t}^4|\mathcal{F}_{t-1})\right)= \sum_{t=1}^{n}\mathbb{E}\left( V_{t}^4 \right) = \sum_{t=1}^{n}\mathbb{E}\left(\sum_{s=1}^{t -1} \varphi_{st} +  \psi_{t}\right)^{4}  =o\left(\operatorname{Var}^{2}(T_{n})\right).
	\end{align}
	By $C_{r}$'s inequality, we obtain
	\begin{align}
	\sum_{t=1}^{n}\mathbb{E}\left( V_{t}^4 \right)
	\leq& 8 \sum_{t=1}^{n} \mathbb{E} \left[(\sum_{s=1}^{t -1} \varphi_{st})^{4} + \psi_{t}^{4}  \right] \notag \\
	=& 8 \sum_{t=1}^{n} \mathbb{E} \left[\sum_{s_{1}, s_{2}, s_{3}, s_{4}=1}^{t -1}  \varphi_{s_{1}t} \varphi_{s_{2}t} \varphi_{s_{3}t} \varphi_{s_{4}t}+ \psi_{t}^{4}  \right]   \notag \\
	=&8 \sum_{t=1}^{n}\left[ 3\sum_{s_{1} \neq  s_{2}}^{t -1} c_{s_{1}t}^{2}c_{s_{2}t}^{2}  \mathbb{E} \mathbf{y}_{t}^{\top} \mathbf{y}_{s_{1}} \mathbf{y}_{s_{1}}^{\top}\mathbf{y}_{t} \mathbf{y}_{t}^{\top}\mathbf{y}_{s_{2}} \mathbf{y}_{s_{2}}^{\top}\mathbf{y}_{t}  + \sum_{s=1}^{t -1} c_{st}^{4}  \mathbb{E}\mathbf{y}_{t}^{\top} \mathbf{y}_{s} \mathbf{y}_{s}^{\top}\mathbf{y}_{t} \mathbf{y}_{t}^{\top}\mathbf{y}_{s} \mathbf{y}_{s}^{\top}\mathbf{y}_{t} \right. \notag \\
	& \left. +   \mathbb{E}  \left(\frac{\beta_{i}^{2}}{n_{i}^2} (\sum_{i=1}^{q} \beta_{i} \boldsymbol{\mu}_{i} )^{\top}  \mathbf{W}^{1/2}\mathbf{y}_{t} \mathbf{y}_{t}^{\top} \mathbf{W}^{1/2} (\sum_{i=1}^{q} \beta_{i} \boldsymbol{\mu}_{i} )\right)^{2}            
	\right] \notag\\
	=&:24 R_{1} + 8R_{2} +  8R_{3},  \label{the spiliting of Expectation of Vn4} 
	\end{align}
	Then, we turn to analyze the following three terms: 
	\begin{align}
	R_{1} =& \sum_{t=1}^{n} \sum_{s_{1} \neq  s_{2}}^{t -1} c_{s_{1}t}^{2}c_{s_{2}t}^{2}  \mathbb{E} \mathbf{y}_{t}^{\top} \mathbf{y}_{s_{1}} \mathbf{y}_{s_{1}}^{\top}\mathbf{y}_{t} \mathbf{y}_{t}^{\top}\mathbf{y}_{s_{2}} \mathbf{y}_{s_{2}}^{\top}\mathbf{y}_{t} \notag \\
	=& \sum_{i=1}^{q}\sum_{t\in \Lambda_{i}} \sum_{s_{1} \neq  s_{2}}^{t -1} O(n^{-8}) \mathbb{E} \mathbf{y}_{t}^{\top} \mathbf{y}_{s_{1}} \mathbf{y}_{s_{1}}^{\top}\mathbf{y}_{t} \mathbf{y}_{t}^{\top}\mathbf{y}_{s_{2}} \mathbf{y}_{s_{2}}^{\top}\mathbf{y}_{t} \quad (\text{by}\quad c_{st} = O(n^{-2}))\notag \\
	\leq & O(n^{-5}) \left[\sum_{i=1}^{q} \operatorname{tr}^{2}( \mathbf{W}\boldsymbol{\Sigma}_{i} )^{2} + \sum_{i< l}^{q} \operatorname{tr}^{2} (\mathbf{W}\boldsymbol{\Sigma}_{i } \mathbf{W}\boldsymbol{\Sigma}_{l})   \right] \notag \\  
	=& o \left( 	\sigma_{n,q_{1}}^{4} \right),   \label{term:R1 the spiliting of Expectation of Vn4} 
	\end{align}
	where the inequality follows from Proposition A.1. in \cite{chen2010tests}. By the same argument, it folows that
	\begin{align}
	R_{2} =& \sum_{t=1}^{n} \sum_{s=1}^{t -1} c_{st}^{4}  \mathbb{E}\mathbf{y}_{t}^{\top} \mathbf{y}_{s} \mathbf{y}_{s}^{\top}\mathbf{y}_{t} \mathbf{y}_{t}^{\top}\mathbf{y}_{s} \mathbf{y}_{s}^{\top}\mathbf{y}_{t}  \notag \\
	=& o \left( 	\sigma_{n,q_{1}}^{4} \right).  \label{term:R2 the spiliting of Expectation of Vn4} 
	\end{align} 
	
	\begin{align}
	R_{3} =&\mathbb{E}\sum_{t=1}^{n}  \left[\frac{\beta_{i}^{2}}{n_{i}^2} (\sum_{i=1}^{q} \beta_{i} \boldsymbol{\mu}_{i} )^{\top}  \mathbf{W}^{1/2}\mathbf{y}_{t} \mathbf{y}_{t}^{\top} \mathbf{W}^{1/2} (\sum_{i=1}^{q} \beta_{i} \boldsymbol{\mu}_{i} ) \right]^{2}    \notag \\
	=& \sum_{l=1}^{q}\sum_{t\in \Lambda_{l}} \mathbb{E} \left[\frac{\beta_{i}^{2}}{n_{i}^2} (\sum_{i=1}^{q} \beta_{i} \boldsymbol{\mu}_{i} )^{\top}  \mathbf{W}^{1/2}\mathbf{y}_{t} \mathbf{y}_{t}^{\top} \mathbf{W}^{1/2} (\sum_{i=1}^{q} \beta_{i} \boldsymbol{\mu}_{i} ) \right]^{2}                                      \notag \\
	= &O(n^{-4}) \left[ (\sum_{i=1}^{q} \beta_{i} \boldsymbol{\mu}_{i} )^{\top}\mathbf{W} \boldsymbol{\Sigma}_{l}  \mathbf{W}  (\sum_{i=1}^{q} \beta_{i} \boldsymbol{\mu}_{i} )\right]^2                \notag       \\
	=& o \left( 	\sigma_{n,q_{2}}^{4} \right), \label{term:R3 the spiliting of Expectation of Vn4}   
	\end{align}
	where the last line we used $\sigma_{n,q_{2}}^{2} = o \left( \sigma_{n,q_{1}}^{2}\right) = o\left( \operatorname{Var} (T_{n})\right)$. By combining  \eqref{the spiliting of Expectation of Vn4}-\eqref{term:R3 the spiliting of Expectation of Vn4} , the proof of \eqref{mean linear :the Lindeberg's condition of CLT} is complete. 
\end{proof}
	Putting together the above results, by the martingale central limit theorem (see Corollary 3.1 of \cite{hall1980martingale}), Theorem \ref{mean linear test via the weighted-norm: main result} can be obtained.

\end{appendices}


\end{document}